\def\draft{0}

\documentclass[11pt]{article}
\usepackage{amsfonts, amsmath, amssymb, url, graphics, amstext, amsthm}
\usepackage{fullpage}

\usepackage[sort,numbers]{natbib}

\usepackage{hyperref}

\newtheorem{theorem}{Theorem}[section]

\newtheorem{definition}[theorem]{Definition}
\newtheorem{lemma}[theorem]{Lemma}
\newtheorem{proposition}[theorem]{Proposition}

\newtheorem{problem}[theorem]{Problem}
\newtheorem{proto}[theorem]{Protocol}

\newtheorem{remk}[theorem]{Remark}

\newcommand{\restateenv}{ZZZ}

\newcommand{\restateenvl}{ZZZ}

\newcommand{\restateenvpr}{ZZZ}

\def\qedsketch{\ifmmode\Box\else{\unskip\nobreak\hfil
\penalty50\hskip1em\null\nobreak\hfil$\Box$
\parfillskip=0pt\finalhyphendemerits=0\endgraf}\fi}

\newcommand{\poly}{{\mathrm{poly}}}

\newcommand{\zo}{\{0,1\}}

\newcommand{\eps}{\varepsilon}

\ifnum\draft=1
\newcommand{\authnote}[2]{{ \bf [#1's Note: #2]}}
\else
\newcommand{\authnote}[2]{}
\fi

\newcommand{\negl}{{\mathrm{neg}}}

\newcommand{\QSZKHV}{\ensuremath{\class{QSZK}_{\class{HV}}}}

\newcommand{\COMMENT}[1]{}
\newcommand{\ket}[1]{|#1\rangle}
\newcommand{\bra}[1]{\langle#1|}
\newcommand{\ketbra}[2]{|#1\rangle\langle#2|}
\newcommand{\altketbra}[1]{\ketbra{#1}{#1}}
\def\01{\{0,1\}}
\newcommand{\braket}[2]{\langle{#1}|{#2}\rangle} 
\def\01{\{0,1\}}

\newcommand{\triple}[3]{\langle{#1}|{#2}|{#3}\rangle}

\newcommand{\Tr}{\mbox{\rm Tr}}

\newcommand{\spa}[1]{\mathcal{#1}}

\newcommand{\tinyspace}{\mspace{1mu}}

\usepackage{dsfont}
\newcommand{\openone}{\mathds{1}}

\newcommand{\tprod}{\otimes}
\newcommand{\smallnorm}[1]{\bigl\lVert\tinyspace#1\tinyspace\bigr\rVert}
\newcommand{\norm}[1]{\left\lVert\tinyspace#1\tinyspace\right\rVert}
\newcommand{\tnorm}[1]{\norm{#1}_{\mathrm{tr}}}
\newcommand{\smtnorm}[1]{\smallnorm{#1}_{\mathrm{tr}}}
\newcommand{\dnorm}[1]{\norm{#1}_{\diamond}}

\newcommand{\abs}[1]{\left\lvert\tinyspace #1 \tinyspace\right\rvert}

\newcommand{\dm}[1]{\dim \mathcal{#1}}

\newcommand{\linear}[1]{\mathbf{L}(\mathcal{#1})}
\newcommand{\density}[1]{\mathbf{D}(\mathcal{#1})}
\newcommand{\unitary}[1]{\mathbf{U}(\mathcal{#1})}
\newcommand{\unitarydim}[1]{\mathbf{U}(#1)}
\newcommand{\sphere}[1]{\mathbf{S}(\mathcal{#1})}
\newcommand{\spheredim}[1]{\mathbf{S}^{#1}}

\newcommand{\id}{\openone}
\newcommand{\identity}[1]{\id_{\mathcal{#1}}}

\newcommand{\tidentity}[1]{\id_{\linear{#1}}}

\newcommand{\tr}{\operatorname{tr}}
\newcommand{\F}{\operatorname{F}}

\newcommand{\ptr}[1]{\tr_\mathcal{#1}}
\newcommand{\ceil}[1]{\left\lceil #1 \right\rceil}

\newcommand{\prob}[1]{\textup{\textsc{#1}}}
\newcommand{\class}[1]{\textup{\textsf{#1}}}

\DeclareMathOperator*{\expect}{\mathbb{E}}

\newenvironment{namedtheorem}[1]
	       {\begin{trivlist}\item {\bf #1.}\em}{\end{trivlist}}

\newcommand{\sbch}{(b_s,h_c)}
\newcommand{\shcb}{(b_c,h_s)}

\title{Quantum Commitments from Complexity Assumptions}
\author{Andr{\'e} Chailloux \\ LRI, Universit{\'e} Paris-Sud \\ \texttt{andre.chailloux@lri.fr}
    \and Iordanis Kerenidis\\ LIAFA, CNRS, Universit{\'e} Paris 7 \\ \texttt{jkeren@liafa.jussieu.fr}
    \and Bill Rosgen \\ CQT, National University of Singapore \\ \texttt{bill.rosgen@nus.edu.sg}}

\date{July 25, 2011}

\begin{document}

\maketitle

\maketitle


\begin{abstract}
 Bit commitment schemes are at the basis of modern cryptography. Since information-theoretic security is impossible both in the classical and the quantum regime, we examine computationally secure commitment schemes. In this paper we study worst-case complexity assumptions that imply quantum
  bit-commitment schemes.  First we show that $\class{QSZK} \not\subseteq
  \class{QMA}$ implies a computationally hiding and statistically
  binding auxiliary-input quantum commitment scheme. We then extend
  our result to show that the much weaker assumption $\class{QIP}
  \not\subseteq \class{QMA}$ (which is weaker than $\class{PSPACE}
  \not\subseteq \class{PP}$) implies the existence of auxiliary-input
  commitment schemes with quantum advice. Finally, to strengthen the
  plausibility of the separation $\class{QSZK} \not\subseteq
  \class{QMA}$  we find a
  quantum oracle relative to which
  honest-verifier \class{QSZK} is not
  contained in \class{QCMA}, the class of languages that can be
  verified using a classical proof in quantum polynomial time. 
\end{abstract}

\section{Introduction}
The goal of modern cryptography is to design protocols that remain
secure under the weakest possible complexity assumptions. Such
fundamental protocols include commitment schemes,
authentication, one-way functions, and pseudorandom generators. 
All these primitives have been shown equivalent: for example
commitment schemes imply one-way functions \cite{IL89} and one-way functions imply commitments \cite{HNORV09,HILL99,N91}.

In this paper we study complexity assumptions that imply commitment
schemes, which are the basis for many cryptographic constructions,
such as zero knowledge protocols for \class{NP}~\cite{BGG+90,GMW91}. A
commitment scheme is a two-phase protocol between a sender and a
receiver. In the commit phase, the sender interacts with the receiver
so that by the end of the phase, the sender is bound to a specific bit, which remains hidden from the receiver until the reveal phase of the protocol, where the receiver learns the bit. 

There are two security conditions for such schemes: binding (the sender cannot reveal more than one value) and hiding (the receiver has no information about the bit before the reveal phase). These conditions can hold statistically, i.e.\ against an unbounded adversary, or computationally, i.e.\ against a polynomial-time adversary.  Without further assumptions these conditions cannot both hold statistically~\cite{LoC97,Mayers97}.

The main complexity assumptions that have been used for the construction of one-way functions, and hence commitments, involve the classes of Computational and Statistical Zero Knowledge. 
Ostrovsky and Wigderson~\cite{OW93} proved that if Computational Zero
Knowledge (\class{ZK}) is not trivial then there exists a family of
functions that are not `easy to invert'. The result was extended by
Vadhan~\cite{V04} to show that if \class{ZK} does not equal
Statistical Zero Knowledge (\class{SZK}), then there exists an
auxiliary-input one-way function, i.e.\ one can construct a one-way
function given an auxiliary input (or else advice).  
Auxiliary-input
cryptographic primitives are natural when considering worst-case
complexity classes: the auxiliary input can encode a
`hard' instance of a problem known only to be hard in the worst
case. 
Last, Ostrovsky and Wigderson also showed that if \class{ZK} contains a `hard-on-average' problem, then `regular' one-way functions exist.

With the advent of quantum computation and cryptography, one needs to
revisit computational security, since many widely-used computational
assumptions, such as the hardness of factoring or the discrete
logarithm problem, 
become false when the adversary is a polynomial-time quantum machine \cite{S94}.    

In this paper, we study worst-case complexity assumptions under which quantum
commitment schemes exist.  As in the classical case, we obtain auxiliary-input
commitments: commitments that can be constructed with classical
and/or quantum advice.  As our commitments are quantum, we define the
computational security properties against quantum poly-time adversaries (who also receive an arbitrary quantum auxiliary input).

Our first result, involves the class of Quantum Statistical Zero Knowledge, \class{QSZK}.

\begin{theorem}\label{thm:qszk-vs-qma-scheme}
If $\class{QSZK} \not\subseteq \class{QMA}$ there exists a
non-interactive auxiliary-input quantum commitment scheme that is statistically-binding and computationally-hiding.
\end{theorem}

Before explaining this result, let us try to see what an equivalent
classical result would mean. At a high level, the classical statement would be of the following form: if \class{SZK} is not in \class{MA}, then auxiliary-input commitments exist. However, under some derandomization assumptions, we have that $\class{NP} = \class{MA} = \class{AM}$ (\cite{MV99,KM99}) and since $\class{SZK} \subseteq \class{AM}$, we conclude that $\class{SZK} \subseteq \class{MA}$. Hence, the equivalent classical assumption is quite strong and, if one believes in derandomization, possibly false.

However, in the quantum setting, it would be surprising if
$\class{QSZK}$ is actually contained in $\class{QMA}$. We know that
$\class{QSZK} \subseteq \class{QIP[2]}$
\cite{Watrous09zero-knowledge}, where $\class{QIP[2]}$ is the class of
languages that have quantum interactive proofs with two messages (note
that one only needs three messages to get the whole power of quantum
interactive proofs).  So far, any attempt to reduce $\class{QIP[2]}$
or $\class{QSZK}$ to $\class{QMA}$ or find any plausible assumptions
that would imply it, have not been fruitful. This seems harder than in
the classical case. The main reason is that the verifier's message
cannot be reduced to a public coin message nor to a pure quantum
state. His message is entangled with his quantum workspace and this
seems inherent for the class $\class{QIP[2]}$ as well as for $\class{QSZK}$. It would be striking if one can get rid of this entanglement and reduce these classes to a single message from the prover.

If we weaken the security condition to hold against quantum adversaries with only classical auxiliary input, then the above assumption also becomes weaker, i.e.\  $\class{QSZK} \not\subseteq \class{QCMA}$,  where $\class{QCMA}$ is the class where the quantum 
verifier receives a single classical message from the prover.  
We give (quantum) oracle evidence for this by showing that
\begin{theorem}\label{thm:oracle-result}\label{thm:oracle-result1}\label{thm:oracle-result2}
There exists a quantum oracle $A$ such that $\QSZKHV^A \not\subseteq \class{QCMA}^A$.
\end{theorem}

\noindent Note that honest-verifier $\QSZKHV=\class{QSZK}$~\cite{Watrous09zero-knowledge} in the unrelativized case.
Our proof of this result extends Aaronson and Kuperberg's
result that there is a quantum oracle $A$ such that
$\class{QMA}^A \not \subseteq \class{QCMA}^A$~\cite{AaronsonK07}.
Subsequent to the completion of this
work, Aaronson has shown the stronger result that there is an
oracle $A$ such that $\class{SZK}^A \not\subseteq
\class{QMA}^A$~\cite{Aaronson11}.  This result implies that our
assumption that $\class{QSZK} \not \subseteq \class{QMA}$ is true relative to an oracle.

We then show the existence of commitment schemes based on a much weaker complexity assumption about
quantum interactive proofs. More precisely, we look at the class
$\class{QIP}$, which was first studied in~\cite{Watrous03}.  This
class is believed to be much larger than $\class{QSZK}$.  We consider
this class and its relation to \class{QMA} to show the following

\begin{theorem}\label{thm:qip-vs-qma}
If $\class{QIP} \not\subseteq \class{QMA}$ there exist non-interactive
auxiliary-input quantum commitment schemes (both statistically hiding
and computationally binding as well as statistically binding and
computationally hiding) with quantum advice.
\end{theorem}

Note, that $\class{QIP} = \class{PSPACE}$ \cite{JJUW10} and $\class{QMA} \subseteq \class{PP}$ \cite{MarriottW05}, so our assumption is extremely weak, in fact weaker than $\class{PSPACE} \not\subseteq \class{PP}$.
Of course, with such a weak assumption we get a weaker form of
commitment: the advice is now quantum.  Thus, in order
for the prover and the verifier to efficiently perform the
commitment for a security parameter $n$, they need to receive a
classical auxiliary input as well as quantum advice of size polynomial
in $n$. This quantum advice is a quantum state on $\poly(n)$ qubits that is not efficiently constructible (otherwise, we could have reduced the quantum advice to classical advice by describing the efficient circuit that produces it). Moreover, the quantum advice we consider does not create entanglement between the players.
 
The key point behind this result is the structure of \class{QIP}. More precisely, we use the fact that there exists a $\class{QIP}$-complete problem where the protocol has only three rounds and the verifier's message is a single coin. The equivalent classical result would say that if three-message protocols with a single coin as a second message are more powerful than $\class{MA}$ then commitments exist. Again, classically, if we believe that $\class{AM} = \class{MA}$, then this assumption is false. Taking this assumption to the quantum realm, it becomes `almost' true, unless $\class{PSPACE}=\class{PP}$.  
  
All of our commitment schemes are non-interactive, a
feature that is useful in many applications.  From
$\class{QIP} \not\subseteq \class{QMA}$ we
construct both statistically hiding and computationally binding
commitments as well as statistically binding and computationally hiding ones,
whose constructions are conceptually different. 
In order to prove the security of the first construction, we prove a parallel repetition theorem for protocols based on the swap test that
may be of independent interest. 
From the $\class{QSZK} \not\subseteq \class{QMA}$ assumption we show
here only statistically binding and computationally
hiding commitments, but computationally binding and statistically
hiding commitments can be similarly shown.

\section{Definitions}\label{sec:defs}

In order to define the statistical distance between quantum states, we
use
the \emph{trace norm}, given by
$\tnorm{X} = \tr \sqrt{X^\dagger X} = \max_{U} \abs{\tr X U}$,
where the maximization is taken over all unitaries of the appropriate
size. Given one of two quantum states
$\rho, \sigma$ with equal probability, the optimal measurement to
distinguish them succeeds with probability
$1/2 + \tnorm{\rho - \sigma}/4$~\cite{Helstrom67}.  Note that this
measurement is not generally efficient.

The \emph{diamond norm} is a generalization of the trace norm to
quantum channels that preserves the distinguishability
characterization.  Given one of two channels $Q_0, Q_1$
with equal probability, then the optimal distinguishing procedure
that uses the channel only once succeeds with probability
$1/2 + \dnorm{Q_0 - Q_1}/4$.
The diamond norm is
more complicated to define than the trace norm, however, as the optimal
distinguishing procedure may need to use an auxiliary space of size
equal to the input space~\cite{Kitaev97,Smith83}.
For a linear map $Q \colon \linear{H} \to
\linear{K}$ with an auxiliary space $\mathcal{F}$ with $\dm{F} =
\dm{H}$, the diamond norm can be defined as
$\dnorm{Q} 
  = \max_{X \in \linear{H \tprod F}} { \tnorm{ Q(X) } } / { \tnorm{X} }.$
One inconvenient property of the diamond norm is that for some maps the maximum in
the definition may not be achieved on a quantum state.  Fortunately,
in the case of the difference of two completely positive maps
this maximum is achieved by a pure state.
\begin{lemma}[\cite{RosgenW05}]\label{diamondtotrace}
Let $\Phi_0, \Phi_1 \colon \linear{H} \to \linear{K}$ be completely
positive linear maps and let $\Phi = \Phi_0 - \Phi_1$. Then, there exists a space $\spa{F}$ and a state $\ket{\phi^*} \in \spa{F}\otimes \spa{H}$ such that
\[ \dnorm{\Phi} = \tnorm{( \tidentity{F}  \otimes \Phi)(\ket{\phi^*}\bra{\phi^*})}. \]
\end{lemma}

Closely related to the diamond norm is a norm studied in operator
theory known as the completely bounded norm.  An upper bound on this
norm can be found
in~\cite{Paulsen02}.   Since the diamond norm is dual to this norm, 
this bound may also be applied to the diamond norm.
See~\cite{JohnstonK+09} 
for a discussion of this bound and the relationship between the
diamond and completely bounded norms.

\begin{lemma}\label{thm:operator-norm-dim-bound}
  Let $\Phi \colon \linear{H} \to \linear{K}$ be a linear map, then
  \begin{equation*}
    \dnorm{\Phi} \leq (\dm{H}) \tnorm{\Phi} 
    = (\dm{H}) \sup_{X \in \linear{H}} \frac{\tnorm{ \Phi(X) }}{\tnorm{X}}.
  \end{equation*}
\end{lemma}

In addition to these norms, we will also make use of the
\emph{fidelity} between two quantum states~\cite{Jozsa94}, which is
given by $\F(\rho,\sigma) = \tr
\sqrt{ \sqrt{\sigma} \rho \sqrt{\sigma}}$.  One property that is important for the results in this
paper is that the fidelity only increases under the application of a
quantum channel.  Specifically, tracing out a portion of two states can only increase their fidelity, i.e.\ for $\rho, \sigma$
density matrices on $\mathcal{H \tprod K}$, it holds that
$\F(\rho, \sigma) \leq \F(\ptr{K} \rho, \ptr{K} \sigma)$.

We also make significant use of the following two properties of the fidelity.
\begin{lemma}[\cite{FuchsG99}]\label{lem:fuchs-van-de-graaf}
  For any density matrices $\rho$ and $\sigma$,
$
    1 - \F(\rho,\sigma)
    \leq \frac{1}{2} \tnorm{\rho - \sigma}
    \leq \sqrt{1 - \F(\rho,\sigma)^2}.
$
\end{lemma}
\begin{lemma}[\cite{SpekkensR01,NayakS03}]\label{lem:sum-fidelity}
  For any density matrices $\rho$ and $\sigma$,
$
    \max_{\xi} \left( \F(\rho, \xi)^2 + \F(\xi, \sigma)^2 \right)
    = 1 + \F(\rho,\sigma).
$
\end{lemma}

\subsection{Quantum Interactive Complexity Classes}

The class \class{QMA}, first studied in~\cite{Watrous00}, is
informally the class of all problems that can be verified by a quantum
polynomial-time algorithm with access to a quantum proof.
\begin{definition}
A language $L$ is in \class{QMA} if there is poly-time
quantum algorithm $V$ (called the \emph{verifier}) such that
\begin{enumerate}
  \item 
if $x \in L$, then there exists a state $\rho$ such that $\Pr[
    V(x,\rho) \text{ accepts}] \geq a$,

  \item if $x \not\in L$, then for any state $\rho$, $\Pr[
    V(x,\rho) \text{ accepts}] \leq b$,
\end{enumerate}
where $a,b$ are any efficiently computable functions of $\abs{x}$ with
$a>b$ with at least an inverse
polynomial gap~\cite{KitaevS+02,MarriottW05}.
If $\rho$ is restricted to be a classical string, the class is called $\class{QCMA}$.
\end{definition}

The class \class{QIP}, first studied in~\cite{Watrous03}, consists of
those problems that can be interactively verified in quantum polynomial time.
A recent result is that
$\class{QIP} = \class{PSPACE}$~\cite{JJUW10}.

\begin{definition}
A language $L \in \class{QIP}$ if there is
a poly-time quantum algorithm $V$
exchanging quantum messages with an unbounded prover $P$
such that for any input $x$
\begin{enumerate}
  \item if $x \in L$ there exists a $P$ such that, 
    $(V,P)$ accepts with probability at least $a$.
  \item if $x \not\in L$, then for any prover $P$, $(V,P)$ accepts with
    probability at most $b$.
\end{enumerate}
As in \class{QMA}, we require only that $a>b$ with at least an inverse
polynomial gap~\cite{KitaevW00}.
\end{definition}

One key property of \class{QIP} is that
any quantum interactive proof system can be simulated by one using
only three messages~\cite{KitaevW00}.  This is not expected to hold in
the classical case, as it would imply that $\class{PSPACE} = \class{AM}$.
This property allows us to define simple problems involving quantum
circuits that are complete for~\class{QIP}.  

In what follows we consider quantum unitary circuits $C$ that
output a state in the space $\spa{O} \otimes \spa{G}$. These spaces
can be different for each circuit. $\spa{O}$ corresponds to the output
space and $\spa{G}$ to the garbage space.
For any circuit $C$, we define $\ket{\phi_C} = C \ket{0}$ in the space
$\spa{O} \otimes \spa{G}$ to be the output of the circuit before the
garbage space is traced out, and $\rho^C =
\Tr_{\spa{G}}(\ketbra{\phi_C}{\phi_C})$ to be the mixed state output
by the circuit after the garbage space is traced out.  
We will also consider more general mixed-state quantum circuits $C$,
that on an input state $\sigma$ and output a quantum state, denoted by
$C(\sigma)$.  Unlike unitary circuits, mixed-state circuits are
allowed to introduce ancillary qubits and trace out qubits during the computation.
Note that circuits of this form can (approximately) represent any quantum channel.
The size of a circuit $C$ is equal to the number of gates in the
circuit plus the number of qubits used by the circuit, denoted
$\abs{C}$.  
We will also use $\abs{\spa{H}}$ to refer to the
size of a Hilbert space $\spa{H}$
i.e.\ $\abs{\spa{H}} = \ceil{\log_2 \dim{\spa{H}}}$.
We use $\linear{H}$ to refer to the set of all linear operators on 
$\mathcal{H}$, and $\density{H}$ to denote the subset of
these operators that are density matrices.
We consider two complete problems for \class{QIP}.

\begin{definition}[\prob{QCD} Problem]
Let $\mu$ be a negligible function.  We define the promise problem
$\prob{QCD}=\{\prob{QCD}_Y,\prob{QCD}_N\}$ with input
two mixed-state quantum circuits $C_0,C_1$ of size $n$ as
\begin{itemize}
\item $(C_0,C_1) \in \prob{QCD}_Y \Leftrightarrow \dnorm{{C_0}-{C_1}} \geq 2- \mu(n)$ 
\item $(C_0,C_1) \in \prob{QCD}_N \Leftrightarrow \dnorm{{C_0}-{C_1}} \leq \mu(n)$
\end{itemize}
\end{definition}

\begin{definition}[$\Pi$ Problem]\label{prob:pi}
Let $\mu$ be a negligible function.  We define the promise problem $\Pi
=\{\Pi_Y,\Pi_N\}$ with input two mixed-state quantum circuits
$C_0,C_1$ of size $n$, where for each $i$ $C_i : \density{X \tprod Y} \to \{0,1\}$, as
\begin{itemize}
\item $(C_0,C_1) \in \Pi_Y \Leftrightarrow \exists \rho^0,\rho^1 \in
  \density{X \tprod Y}$ with $tr_\spa{X}(\rho^0) = tr_\spa{X}(\rho^1)$ such that
\[ \frac{1}{2}\left(\Pr[C_0(\rho^0) = 1] + \Pr[C_1(\rho^1) = 1]\right) = 1 \]
\item $(C_0,C_1) \in \Pi_N \Leftrightarrow \forall \rho^0,\rho^1 \in
  \density{X \tprod Y}$ with $tr_\spa{X}(\rho^0) = tr_\spa{X}(\rho^1)$ we have
\[ \frac{1}{2}\left(\Pr[C_0(\rho^0) = 1] + \Pr[C_1(\rho^1) = 1]\right) \le \frac{1}{2} + \mu(n) \]
\end{itemize}
\end{definition}
\noindent \prob{QCD} is $\class{QIP}$-complete~\cite{RosgenW05}. 
The \class{QIP}-completeness of $\Pi$ follows from a characterization
of \class{QIP} due to Mariott and Watrous~\cite{MarriottW05} that
states that any problem in \class{QIP} has a three message protocol
where the challenge from the Verifier consists of a single coin flip.
We may also assume that this protocol has perfect completeness and soundness
error negligibly larger than 1/2.
Taking the circuits $C_0$ and $C_1$ as the final circuit of the
Verifier in such a proof system when the challenge is either $0$ or
$1$ results in an instance of the problem $\Pi$.  The
\class{QIP}-completeness of $\Pi$ then follows directly from the
completeness and soundness conditions on the proof system.

The complexity class \class{QSZK}, introduced in~\cite{Watrous02}, is
the class of all problems that can be interactively verified by a
quantum verifier who learns nothing beyond the truth of the assertion
being verified.  In the case that the verifier is \emph{honest}, i.e.\
does not deviate from the protocol in an attempt to gain information,
this class can be defined as
\begin{definition}
  A language $L \in \QSZKHV$ if
  \begin{enumerate}
  \item There is a quantum interactive proof system for $L$.
  \item If $x \in L$, the state of the verifier in this proof system after the
    sending of each message can be
    approximated, within negligible trace distance, by a polynomial-time
    preparable quantum state.
    \label{item:zk}
  \end{enumerate}
\end{definition}
If we insist that item~\ref{item:zk} holds when the Verifier
departs from the protocol, the result is the class~\class{QSZK}.
Watrous has shown that 
$\QSZKHV=\class{QSZK}$~\cite{Watrous09zero-knowledge}.
%
This class has complete
problems.
We use the following
$\class{QSZK}$-complete problem~\cite{Watrous02}.
\begin{definition}[\prob{QSD} Problem]
Let $\mu$ be a negligible function. 
$\prob{QSD}=\{\prob{QSD}_Y,\prob{QSD}_N\}$ is the promise problem
on input $(C_0,C_1)$, unitary circuits of size $n$ with $m$
output qubits, such that
\begin{itemize}
\item $(C_0,C_1) \in \prob{QSD}_Y 
  \Leftrightarrow \tnorm{\rho^{C_0}-\rho^{C_1}} \geq 2 - \mu(n)$
\item $(C_0,C_1) \in \prob{QSD}_N \Leftrightarrow \tnorm{\rho^{C_0}-\rho^{C_1}} \leq \mu(n)$
\end{itemize}
\end{definition}


\subsection{Quantum Computational Distinguishability}

The following definitions may be found in~\cite{Watrous09zero-knowledge}.
\begin{definition}\label{defn:comp-distinguishable}
Two mixed states $\rho^0$ and $\rho^1$ on $m$ qubits are
$(s,k,\varepsilon)$-distinguishable if there exists a mixed state
$\sigma$ on $k$ qubits and a quantum circuit $D$ of size $s$ that
performs a two-outcome measurement on $(m+k)$ qubits, such that
$ |\Pr[D(\rho^0\otimes \sigma)=1]-\Pr[D(\rho^1\otimes \sigma)=1] |
\geq \varepsilon. $
If $\rho^0$ and $\rho^1$ are not $(s,k,\varepsilon)$-distinguishable, then they are $(s,k,\varepsilon)$-indistinguishable.
\end{definition}

Let $I\subseteq \zo^*$ and  let an {\em auxiliary-input state
  ensemble} be a collection of mixed states $\{\rho_x\}_{x \in I}$ on
$r(|x|)$ qubits for polynomial $r$ with the property that $\rho_x$ can be efficiently generated given $x$.

\begin{definition}
Two auxiliary-input state ensembles $\{\rho^0_x\}$ and $\{\rho^1_x\}$
on $I$ are {\em quantum computationally indistinguishable} if for all polynomials $p,s,k$ and for all but finitely many $x \in I$, $\rho^0_x$ and $\rho^1_x$ are $(s(|x|),k(|x|),1/p(|x|))$-indistinguishable.
Ensembles $\{\rho^0_x\}$ and $\{\rho^1_x\}$ on $I$ are
{\em quantum computationally distinguishable} if there exist
polynomials $p,s,k$ such that 
for all $x \in I$, $\rho^0_x$ and
$\rho^1_x$ are $(s(|x|),k(|x|),1/p(|x|))$-distinguishable.
\end{definition}
At first glance these definitions of distinguishability and
indistinguishability are not complementary.  We require
distinguishability for all $x \in I$, but require
indistinguishability in only all but finitely many $x \in I$.   This
is because $\abs{x}$ will be our security parameter, and so
while a polynomially-bounded adversary may be able to distinguish the two ensembles for
a finite number of (small) values of $\abs{x}$,
as the parameter grows no efficient algorithm can distinguish the two
ensembles.

Key to this definition is that if two ensembles are computationally distinguishable, then
for all $x$ there exists an efficient procedure in $|x|$ that
distinguishes $\rho^0_x$ and $\rho^1_x$ with probability at least $1/2
+ 1/p(\abs{x})$.  Note that this is not a uniform procedure: the
circuit that distinguishes the two states may depend on $x$.

\begin{definition}
Two auxiliary-input state ensembles $\{\rho^0_x\}$ and $\{\rho^1_x\}$
on $I$ are {\em quantum statistically indistinguishable} if for any polynomial $p$ and for all but finitely many $x \in I$,
$ \tnorm{\rho^0_x - \rho^1_x} \le {1}/{p(|x|)} $. 
\end{definition}

\begin{definition}
Two admissible superoperators $\Phi^0$ and $\Phi^1$ from $t$ qubits to $m$ qubits are $(s,k,\varepsilon)$-distinguishable if there exists a mixed state $\sigma$ on $t+k$ qubits and a quantum circuit $D$ of size $s$ that performs a two-outcome measurement on $(m+k)$ qubits, such that
$ |\Pr[D((\Phi^0 \otimes \id_k) (\sigma))=1]-\Pr[D((\Phi^1 \otimes \id_k) (\sigma))=1] | \geq \varepsilon,$
where $\id_k$ denotes the identity superoperator on $k$ qubits. If the
superoperators $\Phi^0$ and $\Phi^1$ are not
$(s,k,\varepsilon)$-distinguishable, then they are
$(s,k,\varepsilon)$-indistinguishable.
\end{definition}

Let $I\subseteq \zo^*$ and  let an {\em auxiliary-input superoperator
  ensemble} be a collection of superoperators $\{\Phi_x\}_{x \in I}$
from $q(|x|)$ to $r(|x|)$ qubits for some polynomials $q,r$, where as
in the case of states, given $x$ the superoperators can be performed
efficiently in $\abs{x}$.

\begin{definition}
Two auxiliary-input superoperator ensembles $\{\Phi^0_x\}$ and
$\{\Phi^1_x\}$ on $I$ are {\em quantum computationally indistinguishable} if for all polynomials $p,s,k$ and for all but finitely many $x \in I$,  $\Phi^0_x$ and $\Phi^1_x$ are $(s(|x|),k(|x|),1/p(|x|))$-indistinguishable.
Auxiliary-input ensembles $\{\Phi^0_x\}$ and $\{\Phi^1_x\}$
on $I$ are {\em quantum computationally distinguishable} if there
exist polynomials $p,s,k$ such that for all $x \in I$,
$\Phi^0_x$ and $\Phi^1_x$ are
$(s(|x|),k(|x|),1/p(|x|))$-distinguishable.
\end{definition}
If two superoperator ensembles are computationally distinguishable then
there is an efficient (nonuniform) procedure (in $\abs{x}$) to distinguish them
with probability at least $1/2 +
1/p(\abs{x})$ for some polynomial $p$.
If the property of being $(s,k,\varepsilon)$-indistinguishable holds
for all (unbounded) $s$ and all polynomial $k, 1/\varepsilon$, then we
call an ensemble statistically indistinguishable.
Note that these definitions provide a strong quantum analogue of the classical non-uniform notion of computational indistinguishability, since the non-uniformity includes an arbitrary quantum state as advice to the distinguisher. 

We define a new notion that we will use later on. Intuitively, two
circuits that take input in the space $\spa{X}\otimes
\spa{Y}$ and output a single bit are witnessable if there exist two
input states that are identical on $\spa{Y}$ and are accepted
by the two circuits with high probability.

\begin{definition}\label{witnessable}
Two superoperators $\Phi^0$ and $\Phi^1$ from $\linear{X \tprod Y}$ to
a single bit are $(s,k,p)$-witnessable if there exist two input states
$\rho^0, \rho^1 \in \linear{X \tprod Y}$ such that
\begin{enumerate}
\item  $\frac{1}{2}\left(\Pr[\Phi^0(\rho^0) = 1] + \Pr[\Phi^1(\rho^1) =1]\right) \ge 1/2 + 1/p(n)$\label{item:witnessable-1}
\item there exists a state $\sigma \in \linear{W \tprod X
    \tprod Y}$ with $|\spa{W}| =
  k$ and $\ptr{W} \sigma = \rho_0$,
and an admissible superoperator $\Psi : \linear{W \tprod X} \rightarrow \linear{X}$ of size $s$, such that
$ \rho^1 = (\Psi \otimes \tidentity{Y})(\sigma) $\label{item:witnessable-2}
where $\tidentity{Y}$ denotes the identity on $\linear{Y}$.
\end{enumerate}
If 
$\Phi^0$ and $\Phi^1$ are not $(s,k,p)$-witnessable, then they are $(s,k,p)$-unwitnessable.
\end{definition}

\noindent Let $I\subseteq \zo^*$ and  let an {\em auxiliary-input superoperator
  ensemble} be a collection of superoperators $\{\Phi_x\}_{x \in I}$
from $q(|x|)$ to 1 bit for a polynomial $q$, where given $x$ the superoperators can be performed efficiently in $\abs{x}$.

\begin{definition}
Auxiliary-input superoperator ensembles $\{\Phi^0_x\}$ and $\{\Phi^1_x\}$ on $I$ are quantum computationally witnessable if there are polynomials $s,k,p$ such that for all $x \in I$, $\Phi^0_x$ and $\Phi^1_x$ are $(s(|x|),k(|x|),p(|x|))$-witnessable.
Ensembles $\{\Phi^0_x\}$ and $\{\Phi^1_x\}$ on $I$ are quantum computationally unwitnessable if for all polynomials $s,k,p$ and all but finitely many $x \in I$, $\Phi^0_x$ and $\Phi^1_x$ are $(s(|x|),k(|x|),p(|x|))$-unwitnessable.
\end{definition}

\subsection{Quantum Commitments}

\begin{definition}
A quantum commitment scheme (resp.\ with quantum advice) is an interactive protocol $Com=(S,R)$ with the following properties
\begin{itemize}
\item The sender $S$ and the receiver $R$ have common input a security
  parameter $1^n$ (resp.\ both $S$ and $R$ have a copy of a quantum
  state $\ket{\phi}$ of $\poly(n)$ qubits). The sender has private
  input the bit $b \in \zo$ to be committed. Both $S$ and $R$ are
  quantum algorithms that run in time $\poly(n)$ that may exchange
  quantum messages.
\item In the {\em commit} phase,  $S$ interacts with $R$ in order to commit to $b$.
\item In the {\em reveal} phase, $S$ interacts with 
 $R$  in order to reveal $b$.  $R$ decides to accept or reject depending on the revealed value of $b$ and his final state. We say that $S$ reveals $b$, if $R$ accepts the revealed value. In the honest case, $R$ always accepts.
\end{itemize}  

A commitment scheme is {\em non-interactive} if the commit and the
reveal phase each consist of a single message from $S$ to $R$.
When the commit phase is non-interactive, we call $\rho^b_{S}$ the state sent by the honest sender during the commit phase when his bit is $b$. 
\end{definition}

\begin{definition}
A {\em non-interactive auxiliary-input quantum commitment scheme (with
  quantum advice) on} $I$ is a collection of
non-interactive quantum commitment schemes (with advice)
${\cal C}=\{Com_x=(S_x,R_x)\}_{x \in I}$ such that
\begin{itemize}
\item there exists a quantum circuit $Q$ of size polynomial in $|x|$, that given as input $x$ for any $x \in I$, can apply the same maps that $S_x$ and $R_x$ apply during the commitment scheme in time polynomial in $|x|$. 
\item (statistically/computationally hiding) the two auxiliary-input
  state ensembles sent by the honest sender when committing to $0$ or
  $1$, which are given by $\{\rho^0_{S_x}\}_{x\in I}$ and
  $\{\rho^1_{S_x}\}_{x\in I}$, are quantum statistically/computationally indistinguishable.   
\item (statistically/computationally binding) for all but finitely
  many $x \in I$, for all polynomial $p$ and for any
  unbounded/polynomial dishonest senders ${S_{x,0}^*}$, ${S_{x,1}^*}$
  that send the same state in the commit phase
\[ P_{S_x^*}=\frac{1}{2}\left(\Pr[ S_{x,0}^* \mbox{ reveals } b = 0 ] + \Pr[ S_{x,1}^* \mbox{ reveals } b = 1 ]\right) \le \frac{1}{2} + \frac{1}{p(|x|)}
\]
\end{itemize}
\end{definition}

When referring to a commitment scheme, we will use the $\sbch$ and
$\shcb$ to denote schemes that are statistically binding and
computationally hiding and schemes that are computationally binding
and statistically hiding, respectively.

At a high level, the distinction between the two notions, without or
with quantum advice, is the following. We can assume that the two players decide to perform a commitment scheme and agree on a security parameter $n$. Then, in the first case, a trusted party can give them the description of the circuits $(C_0,C_1)$ so that the players can perform the commitment scheme themselves. One can think of the string $(C_0,C_1)$ as classical advice to the players. In the second case, the trusted party gives them the description of the circuits, as well as one copy of a quantum state each. This quantum state is of polynomial size, however it is not efficiently constructible, otherwise the trusted party could have given the players the classical description of the circuit that constructs it. Hence, in the second notion the players receive both classical and quantum advice.


\section{Quantum Commitments Unless \texorpdfstring{$\class{QSZK}
    \subseteq \class{QMA}$}{QSZK is in QMA}}
\label{sec:qszk_vs_qma}

The idea of the proof is to start from pairs of circuits $(C_0,C_1)$
which are in $\prob{QSD}_Y$ which means that their mixed state outputs
$\rho^{C_0}$ and $\rho^{C_1}$ are statistically
far from each other. We want to use $\rho^{C_b}$ as a
commitment state for the bit $b$. Since the states are statistically
far away, such a commitment will be statistically binding. For the
hiding property, we distinguish two cases. If the Receiver can
distinguish in polynomial time (with some quantum auxiliary input) the
two states for all but finitely many such pairs of circuits then we show that $\class{QSZK} \subseteq \class{QMA}$. If the Receiver cannot distinguish the two states for an infinite set I of pairs of circuits, we show how to construct a
non-interactive auxiliary-input quantum $\sbch$-commitment scheme
on $I$. More formally:

\begin{namedtheorem}{Theorem~\ref{thm:qszk-vs-qma-scheme}}
If $\class{QSZK} \not\subseteq \class{QMA}$, then there exists a
non-interactive auxiliary-input quantum $\sbch$-commitment scheme
on an infinite set $I$.
\end{namedtheorem}

\begin{proof}
First, we show the following

\begin{lemma}\label{lem:qszk-ensembles}
If $\class{QSZK} \not\subseteq \class{QMA}$ then there exist two auxiliary-input state ensembles that are quantum computationally indistinguishable on an infinite set $I$. 
\end{lemma}
\begin{proof}
Let us consider the complete problem $\prob{QSD} = \{\prob{QSD}_Y,\prob{QSD}_N\}$ for \QSZKHV.   We may restrict attention to the honest verifier case, since it is known that $\class{QSZK} = \QSZKHV$~\cite{Watrous09zero-knowledge}.  Let $n=|(C_0,C_1)|$ and define $\ket{\phi_{C_b}} = {C_b(\ket{0})}$ in the space $\spa{O} \otimes \spa{G}$ to be the entire output state of the circuit on input $\ket{0}$ and $\rho^{C_b}_{(C_0,C_1)} = \Tr_{\spa{G}}(\ketbra{\phi_{C_b}}{\phi_{C_b}})$ be the output of circuit $C_b$ on $m(n)$ qubits for a polynomial $m$.

Recall that the set $\prob{QSD}_Y$ consists of pairs of circuits
$(C_0,C_1)$, such that the trace norm satisfies
$\smtnorm{\rho^{C_0}_{(C_0,C_1)}-\rho^{C_1}_{(C_0,C_1)}} \geq 2-
\mu(n)$. We now consider the two auxiliary-input state ensembles
$\{\rho^{C_0}_{(C_0,C_1)}\}$ and $\{\rho^{C_1}_{(C_0,C_1)}\}$ for $(C_0,C_1) \in \prob{QSD}_Y$. Assume
for contradiction that they are quantum computationally
distinguishable on $\prob{QSD}_Y$, i.e.\ for some polynomials $p,s,k$
and for all $(C_0,C_1) \in \prob{QSD}_Y$, the states
$\rho^{C_0}_{(C_0,C_1)}$ and $\rho^{C_1}_{(C_0,C_1)}$ are
$(s(n),k(n),1/p(n))$-distinguishable. In other words, for polynomials
$p,s,k$ and for all $(C_0,C_1) \in \prob{QSD}_Y$ there exists a 
state $\sigma$ on $k(n)$ qubits and a quantum circuit $Q$ of size
$s(n)$ that performs a two-outcome measurement on $m(n)+k(n)$
qubits, such that
\[ |\Pr[Q(\rho^{C_0}_{(C_0,C_1)}\otimes \sigma)=1]-\Pr[Q(\rho^{C_1}_{(C_0,C_1)}\otimes \sigma)=1] | \geq \frac{1}{p(n)}.
\]
We now claim that this implies that $\class{QSZK} \subseteq
\class{QMA}$, which is a contradiction. For any input $(C_0,C_1)$ the
prover can send the classical polynomial size description of $Q$ to
the verifier as well as the mixed state $\sigma$ with polynomial
number of qubits. Then, for all $(C_0,C_1) \in \prob{QSD}_Y$, the
verifier with the help of $Q$ and $\sigma$ can distinguish between the
two circuits with probability at least $1/2 + 1/(2p(n))$. On the other hand, for all $(C_0,C_1) \in
\prob{QSD}_N$, no matter what $Q$ and $\sigma$ the prover sends, since
$\smtnorm{\rho^{C_0}_{(C_0,C_1)} -\rho^{C_1}_{(C_0,C_1)}} \leq \mu(n)$ the verifier can only distinguish the two circuits with probability at most $1/2 + \mu(n) / 2$.  This implies that there is an inverse polynomial gap between the acceptance probabilities in the two cases.  By applying standard error reduction tools for \class{QMA}~\cite{KitaevS+02,MarriottW05}, we obtain a $\class{QMA}$ protocol to solve $\prob{QSD}$.

This implies that if $\class{QSZK} \not\subseteq \class{QCMA}$ then
there exists a non empty set $I \subseteq \prob{QSD}_Y$ such that the
two auxiliary-input state ensembles $\{\rho^{C_0}_{(C_0,C_1)}\}$ and
$\{\rho^{C_1}_{(C_0,C_1)}\}$ are quantum computationally
indistinguishable on $I$. Notice that we may take the set $I$ to be
infinite, since if $I$ is finite, then by hard-wiring this finite number of instances into the $\class{QMA}$ verifier (who always accepts these instances), we have again that $\class{QSZK} \subseteq \class{QMA}$.  
 \end{proof}

We now show how to construct a commitment scheme from these ensembles.
\begin{lemma}
The two auxiliary-input state ensembles given by
$\{\rho^{C_0}_{(C_0,C_1)}\}_{(C_0,C_1) \in I}$ and $\{\rho^{C_1}_{(C_0,C_1)}\}_{(C_0,C_1) \in I}$ that are
computationally indistinguishable on the infinite set $I$ imply a  non-interactive auxiliary-input quantum $\sbch$-commitment scheme on $I$.
\end{lemma}
\begin{proof}

For each ${(C_0,C_1)} \in I$ we define a
scheme with security parameter $n=|(C_0,C_1)|$.
\begin{itemize}
\item Commit phase: To commit to bit $b$, the sender $S$ runs the
  quantum circuit $C_b$ with input $\ket{0}$ to create $\ket{\phi_{C_b}}=C_b(\ket{0})$ and sends
  $\rho^{C_b}_{(C_0,C_1)}$ to the receiver $R$, which is the portion
  of $\ket{\phi_{C_b}}$ in the space $\mathcal{O}$.
\item Reveal phase: To reveal bit $b$, the sender $S$ sends the
  remaining qubits of the state $\ket{\phi_{C_b}}$ to the receiver $R$,
  which lie in the space $\mathcal{G}$
 (the honest sender sends $\ket{\phi'} = C_b \ket 0$).
  The receiver applies the circuit $C_b^\dagger$ on his entire state and then measures all his qubits in the computational basis. He accepts if and only if the outcome is $\ket{0}$.
\end{itemize}

Note that all operations of
the sender and the receiver in the above protocol can be computed in
time polynomial in $n$ given the input $(C_0,C_1)$, including the
receiver's test during the reveal phase.  
%
The protocol is computationally hiding since $\{\rho^{C_0}_{(C_0,C_1)}\}$ and $\{\rho^{C_1}_{(C_0,C_1)}\}$ are quantum computationally indistinguishable.

The fact that the protocol is statistically binding follows from the
fact that for the states $\{\rho^{C_0}_{(C_0,C_1)}\}$ and
$\{\rho^{C_1}_{(C_0,C_1)}\}$ (for $(C_0,C_1)\in I \subseteq
\prob{QSD}_Y$) we know that $\|
\rho^{C_0}_{(C_0,C_1)}-\rho^{C_1}_{(C_0,C_1)} \|_{\mathrm{tr}} \geq 2 - \mu(n)$, for a negligible function $\mu$. More precisely, if $\xi$ is the total quantum state sent by a dishonest sender $S^*$ in the commit and reveal phases of the protocol, then the probability that $\xi$ can be revealed as the bit $b$ is 
\begin{equation*}
  \Pr[S^* \mbox{ reveals } b \mbox{ from } \xi]
  = \tr( \ketbra{0}{0} C_b^\dagger \xi C_b )
  = \F( C_b \ket{0}, \xi)^2
  \leq \F( \rho^{C_b}_{(C_0,C_1)}, \ptr{G} \xi )^2
\end{equation*}
using the monotonicity of the fidelity with respect to the partial
trace.  This calculation follows the proof of Watrous that
\class{QSZK} is closed under complementation~\cite{Watrous02}.  In
what follows we consider a dishonest sender that, after the commit
phase, sends one of two different states in the reveal phase, so
the state held by the Receiver is either $\xi_0$ or $\xi_1$.  Notice
that in either case the Sender sends the same state in the commit
phase, so that we have $\ptr{G} \xi_0 = \ptr{G} \xi_1 = \gamma$ 
for some $\gamma \in \density{O}$.  Using this, 
as well as the previous equation and properties of the fidelity
\begin{align*}
  P_{S^*} 
  & =\frac{1}{2}\left(\Pr[S^* \mbox{ reveals } b = 0 \mbox{ from } \xi_0] 
    + \Pr[ S^* \mbox{ reveals } b = 1 \mbox{ from } \xi_1]\right) \\
  & \le \max_{\gamma \in \density{O}} \frac{1}{2} \left(  \F( \rho^{C_0}_{(C_0,C_1)}, \gamma )^2 
    +  \F( \rho^{C_1}_{(C_0,C_1)}, \gamma )^2 \right) \\
  & = \frac{1}{2} \left( 1 + \F(\rho^{C_0}_{(C_0,C_1)},\rho^{C_1}_{(C_0,C_1)}) \right) 
  \le \frac{1}{2} +\frac{\sqrt{\mu(n)}}{2}.
\end{align*}  
The final inequality follows from the relationship between the
fidelity and the trace norm as well as the fact that $\|
\rho^{C_0}_{(C_0,C_1)}-\rho^{C_1}_{(C_0,C_1)} \|_{\mathrm{tr}} \geq 2
- \mu(n)$.  This implies that the protocol is statistically binding.
 \end{proof}
By combining the above Lemmas:  if $\class{QSZK} \not\subseteq
\class{QMA}$, then there exists a
non-interactive auxiliary-input quantum $\sbch$-commitment scheme
on an infinite set $I$.
 \end{proof}

If we are willing to relax the indistinguishability condition, i.e.
enforce the indistinguishability against a quantum algorithm that has only classical auxiliary input (i.e.\ get rid of $\sigma$ in Definition~\ref{defn:comp-distinguishable}),
then the condition becomes $\class{QSZK}
\not\subseteq \class{QCMA}$.  In Section~\ref{sec:oracle} we give
oracle evidence that this
this condition is true.
Notice also that the result of Cr{\'e}peau, L{\'e}gar{\'e}, and
Salvail~\cite{CrepeauLS01} allows this commitment scheme to be used as
a subroutine to
construct a scheme that is statistically hiding and computationally binding.


\section{Quantum \texorpdfstring{$\sbch$}{(bs,hc)}-commitments unless
  \texorpdfstring{$\class{QIP} \subseteq \class{QMA}$}{QIP is in QMA}}\label{sec:qip_vs_qma1}

First, let us note that $\class{QIP} \subseteq \class{QMA}$ implies
that $\class{PSPACE} \subseteq \class{PP}$ which is widely believed
not to be true. Hence, the commitments we exhibit are based on a very
weak assumption.  
Using this weaker assumption, we obtain a weaker commitment scheme,
in the sense that it
requires quantum advice.  Note that our definitions of security are against quantum adversaries that also receive arbitrary quantum advice, hence our honest players are never more powerful than the dishonest ones. Moreover, the quantum advice does not create entanglement between the two players. 

In our first construction, we start from pairs of circuits $(Q_0,Q_1)$
in $\prob{QCD}_Y$ which means that there is a common input $\ket{\phi^*}$ such that their outputs
$\rho^{Q_0}$ and $\rho^{Q_1}$ are statistically
far from each other. We use $\rho^{Q_b}$ as a
commitment state for $b$.  The quantum advice needed for the
commitment is the following: the Sender receives a copy of
$\ket{\phi^*}$ to create the states $\rho^{Q_0}$ and $\rho^{Q_1}$ and
the Receiver also gets a copy of $\ket{\phi^*}$ to check via a SWAP
test that the Sender did not cheat. Using the fact that the states are statistically far apart and a parallel repetition theorem for our swap-test based protocol we obtain negligible binding error.  Similarly to the $\class{QSZK}$ construction, we show that if  $\prob{QCD}$ cannot be solved in $\class{QMA}$ then our scheme is also computationally hiding. 

The remainder of this section provides the proof of this result.
As a first step, we give a scheme with constant binding error
based on the swap test (see~\cite{BuhrmanC+01} for an exposition of the swap test).
  Following this result, we prove a parallel
repetition theorem for non-interactive swap-test based protocols, which
we then use to obtain a scheme with negligible error.

\begin{proposition}\label{prop:qip-vs-qma-1-const-err}
If $\class{QIP} \not\subseteq \class{QMA}$, then there exists a
non-interactive auxiliary-input quantum $\sbch$-commitment scheme with
quantum advice on an infinite set $I$.  This scheme has constant
binding error.
\end{proposition}

\begin{proof}
We first show the following
\begin{lemma}
If $\class{QIP} \not\subseteq \class{QMA}$, there exist two auxiliary-input superoperator ensembles $\{Q^0\}_{(Q^0,Q^1)\in I}$ and $\{Q^1\}_{(Q^0,Q^1)\in I}$ that are quantum computationally indistinguishable on an infinite set $I$.
\end{lemma}
\begin{proof}
Suppose $\class{QIP} \not\subseteq \class{QMA}$.
Let us consider the complete problem $\prob{QCD}$ for $\class{QIP}$ with input the mixed-state circuits $(Q^0,Q^1)$.
Let $n=|(Q^0,Q^1)|$. Let  $\spa{I}$ denote the input space, $\spa{O}$ the output space and 
$\spa{G}$ the output garbage space of the circuits $Q^0,Q^1$.

Consider the set $\prob{QCD}_Y$, whose elements are pairs of circuits
$(Q^0,Q^1)$, such that the diamond norm satisfies $\dnorm{Q^0-Q^1} \geq 2- \mu(n)$, and the two auxiliary-input superoperator ensembles $\{Q^0\}_{(Q^0,Q^1)\in \prob{QCD}_Y}$ and $\{Q^1\}_{(Q^0,Q^1)\in \prob{QCD}_Y}$.
Assume for contradiction that they are quantum computationally distinguishable on $\prob{QCD}_Y$, i.e.\ for some polynomials $p,s,k$ and all $(Q^0,Q^1) \in \prob{QSD}_Y$, the superoperators $Q^0$ and $Q^1$ are $(s(n),k(n),1/p(n))$-distinguishable. In other words, for polynomials $p,s,k$ and for all $(Q^0,Q^1) \in \prob{QSD}_Y$ there exists a mixed state $\sigma$ on $t(n)+k(n)$ qubits and a quantum circuit $D$ of size $s(n)$ that performs a two-outcome measurement on $(m(n)+k(n))$ qubits, such that
\[ |\Pr[D((Q^0 \otimes \id_k) (\sigma))=1]-\Pr[D((Q^1 \otimes \id_k) (\sigma))=1] | \geq \frac{1}{p(n)}
\]
We now claim that this implies that $\class{QIP} \subseteq
\class{QMA}$, which is a contradiction. For any input $(Q^0,Q^1)$ the
$\class{QMA}$-prover can send  to the verifier the classical polynomial size description of $D$ as well as the mixed state $\sigma$ with $\poly(n)$ qubits. Then, for all $(Q^0,Q^1) \in \prob{QCD}_Y$, the verifier with the help of $D$ and $\sigma$ can distinguish between the two circuits with probability higher than $1/2 + 1 / (2p(n))$. On the other hand, for all $(Q^0,Q^1) \in \prob{QCD}_N$, no matter what $D$ and $\sigma$ the prover sends, since  $\dnorm{Q^0-Q^1} \leq \mu(n)$ the verifier can only distinguish the two circuits with probability at most $1/2 + \mu(n)/2$. Hence, there is at least an inverse polynomial gap between the two probabilities, so we can use error reduction~\cite{KitaevS+02,MarriottW05} to obtain a $\class{QMA}$ protocol that solves $\prob{QCD}$ with high probability.

Thus $\class{QIP} \not\subseteq \class{QMA}$
implies that there exists a non-empty set $I \subseteq \prob{QCD}_Y$ and two auxiliary-input superoperator ensembles $\{Q^0\}_{(Q^0,Q^1)\in \prob{QCD}_Y}$ and $\{Q^1\}_{(Q^0,Q^1)\in \prob{QCD}_Y}$ which
 are quantum computationally indistinguishable on $I$.  Once again,
 the set $I$ must be infinite, as if $I$ is finite then by hard-wiring this finite number of instances into the $\class{QMA}$ verifier (who always accepts these instances), we have again that $\class{QIP} \subseteq \class{QMA}$. 
 \end{proof}

We now need to show how to construct a commitment scheme on $I$ based
on these indistinguishable superoperator ensembles.  The protocol we
obtain has only constant binding error: the average of the
probability of successfully revealing 0 and the probability of
successfully revealing 1 is negligibly larger than $3/4$.  Following
this Lemma we prove a parallel repetition result for this protocol
that reduces this error to a negligible function.

\begin{lemma}
The two auxiliary-input superoperator ensembles $\{Q^0\}_{(Q^0,Q^1)\in I}$ and $\{Q^1\}_{(Q^0,Q^1)\in I}$, which are quantum computationally indistinguishable on the infinite set $I \subseteq \prob{QCD}_Y$,
imply a non-interactive auxiliary-input quantum $\sbch$-commitment
scheme with quantum advice on $I$.  This protocol has constant binding error.
\end{lemma}
\begin{proof}
For every ${(Q^0,Q^1)} \in I$ we define a quantum
commitment scheme with quantum advice.  For convenience we let $U^b$ be the unitary operation that simulates the admissible map $Q^b$, in other words we have that $Q^b(\rho) = \tr_{G} U^b (\rho \tprod \ketbra{0}{0}) (U^b)^\dagger$.  Note that any $Q^b$ can be efficiently converted to a unitary circuit $U^b$. Let also $\ket{\phi^*}$ be the pure state from Lemma \ref{diamondtotrace}, such that 
\[ \dnorm{Q^0-Q^1} = \tnorm{( \tidentity{F} \otimes (Q^0-Q^1) )(\ket{\phi^*}\bra{\phi^*})}.
\]
 
\begin{itemize}
\item Define $n=|(Q^0,Q^1)|$ to be the security parameter. $S$ and $R$ also receive as advice a copy of the state $\ket{\phi^*}$ on $\poly(n)$ qubits. 
\item Commit phase: To commit to bit $b$, the sender $S$ runs the
  quantum circuit $\identity{F}\otimes U^b$ with input $\ket{\phi^*}\ket{0}$. The entire output of the circuit is a state in the space $\spa{F} \tprod \spa{O} \tprod \spa{G}$.  
  The sender then sends the qubits in the space $\mathcal{O} \otimes \spa{F}$ to the receiver $R$.
\item Reveal phase: To reveal bit $b$, the sender $S$ sends the
  remaining qubits of the state  $(\identity{F} \tprod U^b)
  (\ket{\phi^*} \ket{0})$  in the space $\spa{G}$ to the receiver $R$. The receiver first applies the operation $\identity{F} \tprod (U^b)^\dagger$ to the entire state he received from the sender and then performs a swap test between this state and his copy of $\ket{\phi^*}\ket 0$.   
\end{itemize}
  
Let us analyze the above scheme. First, note that all operations of
the sender and the receiver in the above protocol can be computed in
time polynomial in $n$ given the input $(Q^0,Q^1)$. This includes the
receiver's test during the reveal phase, since given a description of
a unitary circuit it can be inverted by simply taking the inverse of
each gate and running the circuit in reverse and the swap test is also efficient.

The protocol is computationally hiding since the superoperators $Q^0$ and $Q^1$ are quantum computationally indistinguishable.

The fact that the protocol is statistically binding (with constant error) follows from the fact that
we have $\dnorm{Q^0-Q^1} \geq 2 - \mu(n)$ for a negligible function $\mu$. 
More precisely, let $\sigma^b$ be the state sent by the sender with $\tr_{\spa{G}}\sigma^0 = \tr_{\spa{G}}\sigma^1 = \sigma_{\spa{O}
\spa{F}}$ (the honest sender sends the pure state  $(\identity{F} \tprod U^b) (\ket{\phi^*} \ket 0)$). Then the receiver accepts if and only if the output of $(\identity{F}
  \tprod (U^b)^\dagger) \sigma^b (\identity{F} \tprod U_b)$
  and his copy of $\ket{\phi^*}\ket{0}$ pass the swap test. This probability is equal to
\begin{align*}
  \Pr[S^* \mbox{ reveals } b \mbox{ from } \sigma^b]
  &= \frac{1}{2} + \frac{1}{2} \tr \left[ (\altketbra{\phi^*} \tprod \altketbra 0) (\identity{F}
  \tprod (U^b)^\dagger) \sigma^b (\identity{F} \tprod U_b)\right] \\
  &= \frac{1}{2} + \frac{1}{2} \F( (\identity{F} \tprod U_b) (\altketbra{\phi^*} \tprod \altketbra 0) (\identity{F} \tprod (U^b)^\dagger), \sigma^b )^2 \\
  &\leq \frac{1}{2} + \frac{1}{2} \F( \identity{F} \tprod Q^b (\altketbra{\phi^*}), \ptr{ G } \sigma^b )^2 \\
& \leq \frac{1}{2} + \frac{1}{2} \F( \identity{F} \tprod Q^b (\altketbra{\phi^*}), \sigma_{\spa{O}
\spa{F}} )^2
\end{align*}
where we have used the fact that the swap test on a state $\rho \tprod \sigma$ returns the symmetric outcome with probability $\frac{1}{2} + \frac{1}{2}\tr \rho \sigma$, as well as the monotonicity of the fidelity with respect to the partial trace.

Using this calculation, the binding property of the protocol is given by
\begin{align*}
  P_{S^*} 
  & = \frac{1}{2}\left(\Pr[S^* \mbox{ reveals } b = 0] + \Pr[ S^* \mbox{ reveals } b = 1]\right) \\
  & \leq \frac{1}{2} + \frac{1}{4} \left(  
    \F( \identity{F} \tprod Q^0(\altketbra{\phi^*}), \ptr{ G} \sigma )^2
    +  \F( \identity{F} \tprod Q^1(\altketbra{\phi^*}), \ptr{ G} \sigma )^2 \right) \\
  & \leq \frac{1}{2} + \frac{1}{4} \left( 1 + \F(\identity{F} \tprod Q^0(\altketbra{\phi^*}), \identity{F} \tprod Q^1(\altketbra{\phi^*})) \right) \\
  & \leq \frac{3}{4} + \frac{\sqrt{\mu(n)}}{4},
\end{align*}
where we have used Lemma \ref{diamondtotrace} and Lemma~\ref{lem:sum-fidelity}.
 \end{proof}

From the above two Lemmata, we have that if $\class{QIP}
\not\subseteq \class{QMA}$, then there exists a non-interactive
auxiliary-input quantum $\sbch$-commitment scheme with quantum advice
on an infinite set $I$, with constant binding error. 
 \end{proof}

In the remainder of this section we show how
to reduce the cheating probability of the sender to $1/2 + \negl(n)$.
To do this, we will use parallel repetition of the above protocol.
\begin{proposition}\label{prop:repetition}
Consider a $k$-fold repetition of the above bit commitment
protocol. This is a non-interactive auxiliary-input quantum
$\sbch$-commitment scheme with quantum advice on $I$.
\end{proposition}
\begin{proof}
The two things we have to make sure of is that the computationally hiding property remains under parallel repetition and that the cheating probability of the sender decreases as a negligible function in $k$.
To show that the protocol is computationally hiding, we use the following Lemma.
\begin{lemma}[\cite{Watrous09zero-knowledge}]
Suppose that $\rho_1, \dots \rho_n$ and $\xi_1, \dots, \xi_n$ are m-qubit states such that $\rho_1 \otimes \dots \otimes \rho_n$ and
$\xi_1 \otimes \dots \otimes \xi_n$ are $(s,k,\eps)$-distinguishable. Then there exists at least one choice of $j \in \{1,\dots,n\}$ for which
$\rho_j$ and $\xi_j$ are $(s, (n-1)m + k,\eps/n)$-distinguishable.
\end{lemma}
From this Lemma, we easily have that if the superoperators $Q_0$ and $Q_1$ are quantum computationally indistinguishable 
then the output states of the superoperators $Q_0^{\otimes k}$ and
$Q_1^{\otimes k}$ applied to any product state are quantum
computationally indistinguishable for any $k$ of polynomial size.
This proves that the repeated protocol remains computationally hiding,
since the honest Sender prepares a product state.

We now need to prove that the statistical binding property decreases to
$1/2 + \negl(n)$. We first prove the following Lemma that applies to
the ideal case, i.e.\ the Receiver applies the swap test to one of two
states with orthogonal reduced states.  
The calculation that this strategy (approximately)
generalizes to the case of states that are \emph{almost} orthogonal states follows
the proof of the Lemma.
\begin{lemma}\label{lem:swap-test-repetition}
Let $\ket{\phi_0}, \ket{\phi_1} \in \mathcal{A \tprod B}$ be states
such that $\ptr{B} \altketbra{\phi_0}$ and $\ptr{B}
\altketbra{\phi_1}$ are orthogonal, and let $\rho_0, \rho_1$ be two
states on $(\mathcal{A \tprod B})^{\tprod k} = \mathcal{A}_1 \tprod
\mathcal{B}_1 \tprod \cdots \tprod \mathcal{A}_{k} \tprod \mathcal{B}_{k}$ such that
\begin{equation*}
  \ptr{B_{\mathrm{1}} \tprod \cdots \tprod B_{\mathrm{k}}} \rho_0
  = \ptr{B_{\mathrm{1}} \tprod \cdots \tprod B_{\mathrm{k}}} \rho_1.
\end{equation*}
Consider the following test: \\ \\
\begin{tabular}{l}
Test $b$: Take k copies of $\ket{\phi_b}$ and apply for each $i \in
\{1,\dots,k\}$ the swap test \\ between each copy and the state in
$\spa{A}_i \tprod \spa{B}_i$. Accept if all the swap tests accept.
\end{tabular} \\ \\
For any $\rho_0$ and $\rho_1$ with equal reduced states on
$\spa{A}_1 \tprod \cdots \tprod \spa{A}_k$, we have
\[
\frac{1}{2}\left(\Pr[\rho_0 \mbox{ passes Test } 0] + \Pr[\rho_1 \mbox{ passes Test } 1]\right) \le \frac{1}{2} + \frac{1}{2^{k+1}}
\]
\end{lemma}
\begin{proof}
We prove the result by induction on $k$. For $k = 1$. We have 
\begin{align*}
  \Pr[\rho_b \mbox{ passes Test } b] 
  &= 1/2 + \triple{\phi_b}{\rho_b}{\phi_b}/2 \\
  &= 1/2 + \F(\altketbra{\phi_b}, \rho_b)^2/2\\
  &\leq 1/2 + \F( \ptr{B} \altketbra{\phi_b}, \ptr{B} \rho_b)^2/2.
\end{align*}
Since $\ptr{B} \rho_0 = \ptr{B} \rho_1$, this implies that
\begin{align*}
 \frac{1}{2}&\left(\Pr[\rho_0 \mbox{ passes Test } 0] + \Pr[\rho_1 \mbox{ passes Test } 1] \right)  \\
  &\leq \frac{1}{2} + \frac{1}{4}(
      \F( \ptr{B} \altketbra{\phi_0}, \ptr{B} \rho_0)^2 
      + \F( \ptr{B} \altketbra{\phi_1}, \ptr{B} \rho_1)^2 )\\
  &\leq  \frac{1}{2} + \frac{1}{4} ( 1 + \F( \ptr{B} \altketbra{\phi_0}, \ptr{B} \altketbra{\phi_1}))
  = \frac{3}{4}
\end{align*}
since the reduced states of $\ket{\phi_0},\ket{\phi_1}$ are
orthogonal.

Now we suppose the Lemma is true for $k$ and show it for $k+1$.  
For convenience we set $\spa{S}_i = \spa{A}_i \tprod \spa{B}_i$.  We
take a reference space $\mathcal{R}$ of sufficient size to
consider purifications of $\rho_0$ and $\rho_1$.
Let $\rho_b = \ptr{R} \ketbra{\psi_b}{\psi_b}$ be these (arbitrary) purifications.  
Using this notation, we write
\begin{equation}\label{eqn:lem-defn-psi-0}
  \ket{\psi_0} =
  \alpha_0\ket{\phi_0}_{\spa{S}_1}\ket{\Omega_0}_{\spa{S}_2 \otimes
    \dots \otimes \spa{S}_{k+1} \otimes \spa{R}} + \alpha_1
  \ket{\phi_1}_{\spa{S}_1}\ket{\Omega_1}_{\spa{S}_2 \otimes \dots
    \otimes \spa{S}_{k+1} \otimes \spa{R}} + \alpha_2 \sum_{i = 2}^n
  \ket{\phi_i}\ket{\Omega_i}
\end{equation}
and
\begin{equation}\label{eqn:lem-defn-psi-1}
  \ket{\psi_1} =
  \beta_0\ket{\phi_0}_{\spa{S}_1}\ket{\Gamma_0}_{\spa{S}_2 \otimes \dots
    \otimes \spa{S}_{k+1} \otimes \spa{R}} + \beta_1
  \ket{\phi_1}_{\spa{S}_1}\ket{\Gamma_1}_{\spa{S}_2 \otimes \dots
    \otimes \spa{S}_{k+1} \otimes \spa{R}} + \beta_2 \sum_{i = 2}^n
  \ket{\phi_i}\ket{\Gamma_i}
\end{equation}
where each $\ket{\phi_i},\ket{\phi_j}$ are orthogonal for $i \neq
j$ (for $\ket{\phi_0}$ and $\ket{\phi_1}$ this follows from the fact
that the reduced states on $\spa{A}_1$ are orthogonal). 
Since the goal is to pass swap tests with
$\ket{\phi_0}$ and $\ket{\phi_1}$, we can easily see that we can
take $\alpha_2 = \beta_2 = 0$ without loss of generality, since this state will
only have larger probability of passing the tests.
As one final notational convenience, let $p_i = |\alpha_i|^2$ and $q_i = |\beta_i|^2$.

Before we analyze the probability that the swap tests pass, we show
that the probabilities $p_0$ and $q_1$ satisfy $p_0 + q_1 \leq 1$.  By
Equation~\eqref{eqn:lem-defn-psi-0} we have
\begin{align*}
  p_0 = \abs{\alpha_0}^2
  &= \tr( (\altketbra{\phi_0} \tprod \id) \altketbra{\psi_0}) \\
  &\leq \F(\altketbra{\phi_0}, \tr_{\spa{S}_2 \ldots \spa{S}_{k+1} \spa{R}} \altketbra{\psi_0})^2\\
  &\leq \F(\tr_{\spa{B}_1} \altketbra{\phi_0}, \tr_{\spa{B}_1 \spa{S}_2 \ldots \spa{S}_{k+1} \spa{R} } \altketbra{\psi_0})^2.
\end{align*}
By a similar calculation, we have
\begin{align*}
  q_1 = \abs{\beta_1}^2
 &\leq \F(\tr_{\spa{B}_1} \altketbra{\phi_1}, \tr_{\spa{B}_1 \spa{S}_2
   \ldots \spa{S}_{k+1} \spa{R}} \altketbra{\psi_1})^2.
\end{align*}
Then, using the fact that $\tr_{\spa{B}_1 \spa{S}_2 \ldots
  \spa{S}_{k+1} \spa{R}}
\altketbra{\psi_0} = \tr_{\spa{B}_1 \spa{S}_2 \ldots \spa{S}_{k+1} \spa{R}}
\altketbra{\psi_1}$, as well as the fact that $\tr_{\spa{B}_1}
\altketbra{\phi_0}$ and $\tr_{\spa{B}_1} \altketbra{\phi_1}$ are
orthogonal, we have
\begin{align}
  p_0 + q_1
  &\leq  \F(\tr_{\spa{B}_1} \altketbra{\phi_0}, \tr_{\spa{B}_1
    \spa{S}_2 \ldots \spa{S}_{k+1} \spa{R}} \altketbra{\psi_0})^2
  + \F(\tr_{\spa{B}_1} \altketbra{\phi_1}, \tr_{\spa{B}_1 \spa{S}_2
    \ldots \spa{S}_{k+1} \spa{R}} \altketbra{\psi_1})^2 \nonumber \\
  & \leq 1 + \F(\tr_{\spa{B}_1} \altketbra{\phi_0}, \tr_{\spa{B}_1}
  \altketbra{\phi_1}) \nonumber \\
  &= 1. \label{eqn:lem-p0-q1}
\end{align}

We now analyze the probability that the swap tests pass.
Consider applying test $0$ on $\ket{\psi_0}$. When applying the swap
test between $\ket{\phi_0}$ and $\ket{\phi_0}$, the result is the
state $\ket{0}\ket{\phi_0}\ket{\phi_0}$ where the first register
corresponds to the acceptance of the swap test ($0$ corresponds to
accept). When applying the swap test between the two states $\ket{\phi_0}$ and
$\ket{\phi_1}$, the result before measuring the first qubit is \[\frac{1}{\sqrt{2}}\left(\ket{0}(\ket{\phi_0}\ket{\phi_1} + \ket{\phi_1}\ket{\phi_0}) + \ket{1}(\ket{\phi_0}\ket{\phi_1} - \ket{\phi_1}\ket{\phi_0})\right). \]
So the swap test on the space $\spa{S}_1$ accepts with probability $p_0 + p_1/2$. Conditioned on this test passing, we have the state:
\[
\frac{1}{\sqrt{p_0 + p_1/2}} \left[
  \alpha_0\ket{\phi_0}\ket{\phi_0}\ket{\Omega_0}_{\spa{S}_2 \otimes
    \dots \otimes \spa{S}_{k+1} \spa{R}} + \frac{\alpha_1}{\sqrt{2}}
  (\ket{\phi_0}\ket{\phi_1} +
  \ket{\phi_1}\ket{\phi_0})\ket{\Omega_1}_{\spa{S}_2 \otimes \dots
    \otimes \spa{S}_{k+1} \spa{R}} \right] \]
Discarding the first system results in the state in $\spa{S}_2 \otimes
\dots \otimes \spa{S}_{k+1} \otimes \spa{R}$ 
(using orthogonality of $\ket{\phi_0}$ and $\ket{\phi_1}$) given by
\[ \sigma = \frac{p_0}{p_0 + \frac{p_1}{2}} \ketbra{\Omega_0}{\Omega_0} + \frac{\frac{p_1}{2}}{p_0 + \frac{p_1}{2}} \ketbra{\Omega_1}{\Omega_1}
\]
Let $T_0(\xi)$ be the probability that a state $\xi \in \spa{S}_2
\otimes \dots \otimes \spa{S}_{k+1} \otimes \spa{R}$ passes all swap
tests in $\spa{S}_2 \otimes \dots \otimes \spa{S}_{k+1}$ with
$\ket{\phi_0}$.  We include the space $\spa{R}$ for convenience only:
notice that the choice of purification in the space $\spa{R}$ has no
effect on this probability.  Using this notation, we have
\begin{eqnarray*}
\Pr[\rho_0 \mbox{ passes Test } 0] & = & (p_0 + \frac{p_1}{2})\cdot \left(\frac{p_0}{p_0 + \frac{p_1}{2}}T_0(\ketbra{\Omega_0}{\Omega_0}) + \frac{\frac{p_1}{2}}{p_0 + \frac{p_1}{2}}T_0(\ketbra{\Omega_1}{\Omega_1})\right) \\
& = & p_0 T_0(\ketbra{\Omega_0}{\Omega_0}) + \frac{p_1}{2} T_0(\ketbra{\Omega_1}{\Omega_1})
\end{eqnarray*}
Similarly, we define $T_1(\xi)$ for any $\xi$ and we have 
\[
\Pr[\rho_1 \mbox{ passes Test } 1] = \frac{q_0}{2} T_1(\ketbra{\Gamma_0}{\Gamma_0}) + q_1 T_1(\ketbra{\Gamma_1}{\Gamma_1})
\]
which gives us
\begin{align}
P & = \frac{1}{2}\left(\Pr[\rho_0 \mbox{ passes Test } 0] + \Pr[\rho_1
  \mbox{ passes Test } 1]\right)\nonumber \\  
& = \frac{1}{2}\left(p_0 T_0(\altketbra{\Omega_0}) + \frac{p_1}{2} T_0(\altketbra{\Omega_1}) + 
\frac{q_0}{2} T_1(\altketbra{\Gamma_0}) + q_1
T_1(\altketbra{\Omega_1})\right) \label{eqn:lem-prob-P}
\end{align}
Consider the states $\xi_0 = p_0 \altketbra{\Omega_0} + p_1
\altketbra{\Omega_1}$ and $\xi_1 = q_0 \altketbra{\Gamma_0} + q_1
\altketbra{\Gamma_1}$.  These states are obtained from $\rho_0$ and
$\rho_1$ by discarding the system in $\spa{S}_1$.  This implies that
they have the properties in the statement of the Lemma, i.e.\ the
reduced states of $\xi_0$ and $x_1$ on $\spa{A}_2 \tprod \cdots \tprod
\spa{A}_{k+1}$ are equal.  Thus,
by induction, we know that $\frac{1}{2}\left(T_0(\xi_0) + T_1(\xi_1)\right) \le \frac{1}{2} + \frac{1}{2^{k+1}}$. This means that:
\begin{equation*}
  \frac{1}{2}\left(p_0 T_0(\altketbra{\Omega_0}) + p_1 T_0(\altketbra{\Omega_1}) + 
q_0 T_1(\altketbra{\Gamma_0}) + q_1
T_1(\altketbra{\Gamma_1})\right) \le \frac{1}{2} +
\frac{1}{2^{k+1}}
\end{equation*}
Using this, as well as Equation~\eqref{eqn:lem-prob-P}, we have 
\begin{eqnarray*}
P & = & \frac{1}{2}\left(p_0 T_0(\altketbra{\Omega_0}) + \frac{p_1}{2} T_0(\altketbra{\Omega_1}) + 
\frac{q_0}{2} T_1(\altketbra{\Gamma_0}) + q_1 T_1(\altketbra{\Gamma_1})\right) \\
& = & \frac{1}{4} + \frac{1}{2^{k+2}} + \frac{p_0}{4} T_0(\altketbra{\Omega_0}) + \frac{q_1}{4} T_1(\altketbra{\Gamma_1}) \\
& \le & \frac{1}{2} + \frac{1}{2^{k+2}}, 
\end{eqnarray*}
where the final inequality is by Equation~\eqref{eqn:lem-p0-q1}.
 \end{proof}

Notice that in the original bit commitment protocol the Receiver
applies the swap test to $\ket{\phi^*} \ket 0$ and the output of
$(U_b^\dagger \tprod \id)(\sigma_b)(U_b \tprod \id)$ where $\sigma_b$ is the state sent during
the protocol.  Since $U_b^\dagger$ is unitary, this is equivalent to
applying the swap test between $\sigma_b$ and the state $\ket{\phi_b} = (U_b \tprod \id)
\ket{\phi^*}\ket 0$, for whatever value of $b$ the Sender has
revealed.  Viewed in this way, the receiver applies the swap test
between $\sigma_b$ and one of two \emph{almost} orthogonal states.
Furthermore, these two states have the property that the reduced
states on the space $\spa{O}$ have negligible fidelity.
Notice also that the Sender may send one of two states
$\sigma_0$ and $\sigma_1$ depending on the value that he wishes to
reveal.  Since we are interested in the sum of the probabilities that
the Sender can successfully reveal both 0 and 1 in a given instance of
the protocol, we may assume that the first message stays the same,
i.e.\ that $\ptr{G} \sigma_0 = \ptr{G} \sigma_1$.  This is exactly the
condition in Lemma~\ref{lem:swap-test-repetition} with the exception
that instead of the orthogonality of the states $\ket{\phi_i}$ we have
only approximate orthogonality.  We are able to overcome this obstacle
with the following Lemma, the proof of which makes significant use of
the fact that the trace norm can be written in
terms of the projectors onto the positive and negative eigenspaces of
a matrix.  In particular, when
applied to a Hermitian operator $X$ the trace norm is given by $\tr
(\Pi_+ X) - \tr (\Pi_- X)$, where $\Pi_+$ and $\Pi_-$ are the
projectors onto the positive and negative eigenspaces of $X$,
respectively.  This fact follows from the definition of the trace norm.
\begin{lemma}\label{lem:approx-orthogonal}
  Let $\ket{ \phi_0 }, \ket{ \phi_1} \in \spa{A \tprod B}$ such that
  $\tnorm{ \ptr{B} \altketbra{\phi_0} - \ptr{B} \altketbra{\phi_1} } \geq
  2 - \varepsilon$.  Then there exist states $\ket{ \phi'_0 }, \ket{
    \phi'_1} \in \spa{A \tprod B}$ such that
  \begin{enumerate}
  \item $\braket{\phi'_i}{\phi_i} \geq 1 - \varepsilon$ for $i \in \{0,1\}$,
  \item $\ptr{B} \altketbra{\phi'_0}$ and $\ptr{B}
    \altketbra{\phi'_1}$ are orthogonal.
  \end{enumerate}
\end{lemma}
\begin{proof}
    For simplicity, let $\rho_i = \ptr{B} \altketbra{\phi_i}$.  We have
    \begin{equation}\label{eqn:pos-neg-tnorm}
      2 - \varepsilon
      \leq \tnorm{ \rho_0 - \rho_1 } 
      = \tr \abs{\rho_0 - \rho_1}
      = \tr \Pi_+ (\rho_0 - \rho_1) - \tr \Pi_- (\rho_0 - \rho_1),
    \end{equation}
    where $\Pi_+$ and $\Pi_-$ are the projectors onto the positive and
    negative eigenspaces of $\rho_0 - \rho_1$ respectively.  Notice
    that
    \begin{equation*}
      \tr (\Pi_+ \rho_0)
      = \tr (\Pi_+ (\rho_0 - \rho_1)) + \tr (\Pi_+ \rho_1)
      \geq \tr (\Pi_+ (\rho_0 - \rho_1)),
    \end{equation*}
    and similarly $\tr (\Pi_- \rho_1) \geq - \tr (\Pi_- (\rho_0 -
    \rho_1))$, which implies that
    \begin{equation*}
      \tr (\Pi_+ \rho_0) + \tr (\Pi_- \rho_1)
      \geq \tr (\Pi_+ (\rho_0 - \rho_1)) - \tr (\Pi_- (\rho_0 - \rho_1))
      \geq 2 - \varepsilon,
    \end{equation*}
    by Equation~\eqref{eqn:pos-neg-tnorm}.  This implies that $\tr
    (\Pi_+ \rho_0) \geq 1 - \varepsilon$ and $\tr (\Pi_- \rho_1) \geq 1 - \varepsilon$.

    We introduce the states $\rho'_i$ given by the (renormalized)
    projection of $\rho_0$ and $\rho_1$ into the spaces spanned by
    $\Pi_+$ and $\Pi_-$, respectively.  Since these are orthogonal
    projectors the states $\rho'_0$ and $\rho'_1$ are orthogonal.
    Notice also that
    \begin{equation*}
      \tnorm{ \rho_0 - \rho'_0 } 
      = \tr \abs{ \rho_0 - \rho'_0 }
      = \tr (\Gamma_+ (\rho_0 - \rho'_0)) - \tr (\Gamma_-(\rho_0 - \rho'_0))
      = 2 \tr (\Gamma_+(\rho_0 - \rho'_0)),
    \end{equation*}
    where $\Gamma_+, \Gamma_-$ are the projectors onto the positive
    and negative eigenspaces of $\rho_0 - \rho'_0$, and we have also
    used the fact that $\tr(\rho_0 - \rho'_0) = 0$, which implies that
    the positive portion of $\rho_0 - \rho'_0$ has the same trace as the
    negative portion.  Consider the
    positive eigenspace of $\rho_0 - \rho'_0$.  This is precisely the
    subspace spanned by the support of $\rho_0$ that lies outside the
    support of $\rho'_0$, i.e.\ this is exactly the space spanned by
    the projector $\Pi_- = \Gamma_+$.  Using this observation
    \begin{equation}\label{eqn:rho-0}
        \tnorm{ \rho_0 - \rho'_0 } 
        = 2 \tr (\Gamma_+(\rho_0 - \rho'_0))
        = 2 \tr (\Pi_- \rho_0)
        \leq 2 \varepsilon,
    \end{equation}
    where we have used the fact that $\tr (\Pi_- \rho_0) = 1 - \tr (\Pi_+
    \rho_0) \leq \varepsilon$.
    A similar argument establishes the fact that
    \begin{equation}\label{eqn:rho-1}
        \tnorm{ \rho_1 - \rho'_1 } 
        = 2 \tr (\Pi_+ \rho_1)
        \leq 2 \varepsilon.
    \end{equation}

    Finally, we note that Equations~\eqref{eqn:rho-0}
    and~\eqref{eqn:rho-1} and Uhlmann's theorem imply that there exist
    purifications $\ket{\phi'_0}, \ket{\phi'_1} \in \mathcal{A \tprod
      B}$ of $\rho'_0$ and $\rho'_1$ such that
    \begin{equation*}
      \braket{\phi'_i}{\phi_i}
      = \F( \rho'_i, \rho_i )
      \geq 1 - \varepsilon.
    \end{equation*}
    This, combined with the orthogonality of $\rho'_0$ and $\rho'_1$,
    completes the proof.
   \end{proof}

This Lemma shows that we may replace the two states that are almost
orthogonal with nearby states that have exactly the orthogonality
property required by Lemma~\ref{lem:swap-test-repetition}, which we
can in turn use to show that the protocol repeated $k$
times is statistically binding.
To do so, notice that the two states $\ket{\phi_0}$ and
$\ket{\phi_1}$, which are given by applying the circuits $Q_0$ and
$Q_1$ to the state $\ket{\phi^*} \ket 0$, satisfy
\begin{align*}
  \tnorm{ \altketbra{\phi_0} - \altketbra{\phi_1} }
  & \geq \tnorm{ \ptr{G} (\altketbra{\phi_0} - \altketbra{\phi_1}) } \\
  & = \tnorm{ ((Q_0 - Q_1) \tprod I)(\altketbra{\psi^*})} \\
  & = \dnorm{Q_0 - Q_1} \\
  & \geq 2 - \mu(n),
\end{align*}
These states are not orthogonal, but are nearly so.  We may,
however, use Lemma~\ref{lem:approx-orthogonal} to obtain
$\ket{\phi'_0}$ and $\ket{\phi'_1}$ that have the orthogonality
property required by Lemma~\ref{lem:swap-test-repetition} that have
inner product at least $1 - \mu(n)$ with the original states
$\ket{\phi_0}$ and $\ket{\phi_1}$, respectively.  

We now relate the probability that the state $\rho$ passes our Test 0,
i.e.\ the $k$ swap tests with the state $\ket{\phi_0}^{\otimes k}$ to
the probability that the same state $\rho$ passes the $k$ swap tests
with the state $\ket{\phi'_0}^{\otimes k}$ (denoted by
$\mbox{Test}^\prime$ 0). The difference of these probabilities is
upper bounded by the trace distance of the difference of the states
$\ket{\phi_0}^{\otimes k}$ and $\ket{\phi'_0}^{ \otimes k}$, since we
can view the swap test with $\rho$ as a measurement to distinguish
these two states.  This gives
\begin{eqnarray*}
| \Pr[\rho \mbox{ passes Test } 0]  - \Pr[\rho \mbox{ passes Test}' \;0]  |  & \leq &
 \tnorm{ (\altketbra{\phi_0})^{\otimes k} - (\altketbra{\phi'_0})^{\otimes k} } \\
& = & 2 \sqrt{1-\abs{\braket{ \phi'_0 }{\phi_0}}^{2k}} \\
& \leq & 2 \sqrt{1-(1-\mu(n))^{2k}} \\
& \leq & 2 \sqrt{2 k \mu(n)},
\end{eqnarray*}
where the final inequality is Bernoulli's inequality.
Similarly we have
\[
| \Pr[\rho \mbox{ passes Test } 1]  - \Pr[\rho \mbox{ passes Test}'\;1]  |  \leq 
    2 \sqrt{2 k \mu(n)} 
\]
Hence, for the binding property of our scheme we have
\begin{eqnarray*}
\lefteqn{\frac{1}{2}\left(  \Pr[\rho \mbox{ passes Test } 0] +  \Pr[\rho \mbox{ passes Test } 1] \right)}\\
& \leq & \frac{1}{2}\left(  \Pr[\rho \mbox{ passes Test}'\; 0] +  \Pr[\rho \mbox{ passes Test}'\; 1] \right) + 2 \sqrt{2 k \mu(n)}\\
& \le & \frac{1}{2} + \frac{1}{2^{k+1}} +  2 \sqrt{2 k \mu(n)}.
\end{eqnarray*}
since, for the $\mbox{Test}'\;0$ and $\mbox{Test}'\;1$ we can use
Lemma~\ref{lem:swap-test-repetition} for the perfect case.  This
quantity is negligibly larger than $1/2$, as we may take $k$ any
polynomial and $\mu$ is a negligible function.
 \end{proof}

This proposition, when combined with
Proposition~\ref{prop:qip-vs-qma-1-const-err}, gives the main result
of this section.

\begin{theorem}\label{thm:qip-vs-qma-1}
If $\class{QIP} \not\subseteq \class{QMA}$, then there exists a
non-interactive auxiliary-input quantum $\sbch$-commitment scheme with
quantum advice on an infinite set $I$.
\end{theorem}


\section{Quantum \texorpdfstring{$\shcb$}{(hs,bc)}-commitments unless
  \texorpdfstring{$\class{QIP} \subseteq \class{QMA}$}{QIP is not in QMA}}\label{sec:qip-scheme-2}

To obtain protocols that are computationally binding and statistically
hiding, we use instances of the \class{QIP}-complete problem
$\Pi$ to construct a $\shcb$-commitment scheme
with quantum advice under the assumption that $\class{QIP}
\not\subseteq \class{QMA}$. We start from pairs of circuits
$Q_0,Q_1 \in \Pi_Y$ and the corresponding input states $\rho^0,\rho^1$
(see Definition~\ref{prob:pi}) that will be given to the Sender as
quantum advice. An honest Sender commits to $b$ by sending half of
$\rho^b$ to the Receiver. By definition of $\rho^0,\rho^1$, the
protocol is statistically hiding (in fact it is perfectly hiding). 
During the reveal phase, the Sender sends the second half of $\rho^b$. If $\Pi \not\in \class{QMA}$, we show that this protocol is also computationally binding, using our notion of computationally unwitnessable superoperators. 

\begin{theorem}
If $\class{QIP} \not\subseteq \class{QMA}$, then there exists a non-interactive auxiliary-input quantum $\shcb$-commitment scheme with quantum advice on an infinite set $I$.
\end{theorem}
\begin{proof}
Recall the Complete problem $\Pi = \{\Pi_Y,\Pi_N\}$ from
Definition~\ref{prob:pi} with inputs the mixed-state circuits
$(Q^0,Q^1)$ from $\density{X \tprod Y}$ to a single bit and
$n=|(Q^0,Q^1)|$. To show this Theorem, we use the following Lemma, the
proof of which is very similar to the proof of Lemma~\ref{lem:qszk-ensembles}.

\begin{lemma}\label{lem:qip-ensembles}
If $\class{QIP} \not\subseteq \class{QMA}$, there exist two auxiliary-input superoperator ensembles $\{Q^0\}_{(Q^0,Q^1) \in I}$ and   
$\{Q^1\}_{(Q^0,Q^1) \in I}$ that are quantum computationally unwitnessable on an infinite set $I$.
\end{lemma}
\begin{proof}
Let us consider the set $\Pi_Y$ and suppose for contradiction that the two auxiliary-input superoperator ensembles $\{Q^0\}_{(Q^0,Q^1) \in \Pi_Y}$ and   
$\{Q^1\}_{(Q^0,Q^1) \in \Pi_Y}$ are quantum computationally
witnessable, i.e.\ there exist polynomials $(s,k,p)$ such that for all
$(Q^0,Q^1) \in \Pi_Y$ the superoperators $Q^0$ and $Q^1$ are
$(s(n),k(n),p(n))$-witnessable. In other words,  there exist polynomials
$(s,k,p)$ such that for all $(Q^0,Q^1) \in \Pi_Y$ there exist two
input states $\rho^0, \rho^1 \in \linear{X \tprod Y}$ such that first,
there exists a state $\sigma \in \linear{W \tprod X \tprod Y}$ with
$|\spa{W}| = k$ and $\ptr{W} \sigma = \rho_0$, and there exists
an admissible superoperator $\Psi : \linear{W \tprod X} \rightarrow \linear{X}$ of size $s$, such that $\rho^1 = (\Psi \otimes \identity{Y})(\sigma)$; and second
\[
\frac{1}{2}\left(\Pr[Q^0(\rho^0) = 1] + \Pr[Q^1(\rho^1) =1]\right) \ge \frac{1}{2} + \frac{1}{p(n)}.\]

Then, we provide a \class{QMA} protocol for the problem $\Pi$. 
Merlin sends $\sigma$ (which is of size polynomial in the input, since $k(n) = |\spa{W}|$) and the classical
description of $\Psi$ (of size $s(n)$).
Arthur with probability $1/2$ applies $Q^0$ on $\rho^0$ (which he
obtains from $\sigma$ by discarding the space $\spa{W}$) and accepts if he gets $1$; and with probability $1/2$ he first creates $\rho^1$ from $\Psi$ and $\sigma$, then applies $Q^1$ on it and also accepts if he gets $1$. \\
(\emph{Completeness}) If $(Q^0,Q^1) \in \Pi_Y$, we have
\[
\Pr[\mbox{Arthur accepts}]  =  \frac{1}{2}\left(\Pr[Q^0(\rho^0) = 1] + 
\Pr[Q^1(\rho^1) = 1]\right) \ge  \frac{1}{2} + \frac{1}{p(n)}
\] 
 (\emph{Soundness}) If $(Q^0,Q^1) \in \Pi_N$, then for any cheating
 Merlin, Arthur receives a state $\rho^0_*$, from which he constructs
 (with half probability) a state $\rho^1_*$ each in space $\spa{X}
 \otimes \spa{Y}$ such that $\tr_{\spa{X}} \rho^0_* = \tr_{\spa{X}}
 \rho^1_*$. By the definition of $\Pi_N$, we have
\[
\Pr[\mbox{Arthur accepts}]  =  \frac{1}{2}\left(\Pr[Q^0(\rho^0_*) = 1] +
\Pr[Q^1(\rho^1_*) = 1]\right) \leq  \frac{1}{2} + \mu(n)
\]
We have an inverse polynomial gap between completeness and soundness and hence we conclude that $\Pi \in \class{QMA}$.
This proves that there is a nonempty $I$ that satisfies the property
of our Lemma. Note that if $I$ is finite, then by hard-wiring this
finite number of instances into the $\class{QMA}$ verifier (who always
accepts these instances), we have again that $\class{QIP} \subseteq
\class{QMA}$. So if $\class{QIP} \not\subseteq \class{QMA}$ then the
set $I$ can be taken to be infinite. 
 \end{proof}

\noindent To finish the proof of the Theorem, we now need to show the following.

\begin{lemma}
Auxiliary-input superoperator ensembles $\{Q^0\}_{(Q^0,Q^1) \in I}$ and 
$\{Q^1\}_{(Q^0,Q^1) \in I}$ that are quantum computationally unwitnessable on an infinite set $I \subseteq \Pi_Y$ imply a non-interactive quantum $\shcb$-commitment scheme with quantum advice on $I$.
\end{lemma}

\begin{proof}
\emph{Commitment scheme}
Each $(Q^0,Q^1) \in I \subseteq \Pi_Y$ gives the following scheme
\begin{itemize}
\item Let $n=|(Q^0,Q^1)|$ be the security parameter. The sender
  receives as advice $\rho^0,\rho^1 \in \spa{X}^i \otimes \spa{Y}^i$
  such that
$\tr_{\spa{X}}\rho^0 = \tr_{\spa{X}}\rho^1$
and
$\frac{1}{2}\left(\Pr[Q^0(\rho^0) = 1] + \Pr[Q^1(\rho^1) = 1]\right) \ge 1  - \mu(n)$. 
For consistency with our definitions, we also suppose that the Receiver gets a copy of $\rho^0,\rho^1$. These states will not be used in the honest case and they will not harm the security for a cheating Receiver.
\item (Commit phase) To commit to $b$, the Sender sends the state in $\spa{Y}^b$ to the Receiver.
\item (Reveal phase) To reveal $b$, the Sender sends the state in $\spa{X}^b$. The Receiver applies $Q^b$ on the space $\spa{X}^b \otimes \spa{Y}^b$ and accepts if he gets $1$.
\end{itemize}
\emph{Statistical hiding property}: The states that the receiver gets in the commit phase satisfy $\tr_{\spa{X}}\rho^0 = \tr_{\spa{X}}\rho^1$ and hence our scheme is perfectly hiding.

\noindent
\emph{Computationally binding property}:
The property follows from the fact that the two auxiliary-input superoperator ensembles $\{Q^0\}_{(Q^0,Q^1) \in I}$ and 
$\{Q^1\}_{(Q^0,Q^1) \in I}$ are quantum computationally unwitnessable.
Fix $(Q^0,Q^1) \in I$ with $|(Q^0,Q^1)|=n$. After the reveal phase, the Receiver has ${\rho}_*^b$ in space $\spa{X} \otimes \spa{Y}$, where $b$ is the revealed bit. 
Since we consider dishonest senders $S^*_{(Q^0,Q^1)}$ that are quantum
polynomial time machines with quantum advice, the states ${\rho}_*^0$
and ${\rho}_*^1$ satisfy property~\ref{item:witnessable-2} of
Definition \ref{witnessable}. Thus, for all but finitely many
$(Q^0,Q^1) \in I$ they do not have property~\ref{item:witnessable-1} of Definition \ref{witnessable}. Then, for such $(Q^0,Q^1) \in I$ we have 
\begin{align*}
P_{S^*_{(Q^0,Q^1)}} & =  \frac{1}{2}\left(\Pr[ S_{(Q^0,Q^1)}^* \mbox{ reveals } b = 0 ] + \Pr[ S_{(Q^0,Q^1)}^* \mbox{ reveals } b = 1 ]\right) \\
& = \frac{1}{2}\left(\Pr[Q_0({\rho}_*^0) = 1] + \Pr[Q_1({\rho}_*^1) = 1] \right) 
 \le  \frac{1}{2} + \frac{1}{p(n)}
\end{align*}
for all polynomials $p$ 
\end{proof}
\noindent From the above two Lemmas, unless $\class{QIP}
\subseteq \class{QMA}$ there exists a non-interactive auxiliary-input
quantum $\shcb$-commitment scheme with quantum advice on infinite set $I$.
\end{proof}

\noindent This result, combined with Theorem~\ref{thm:qip-vs-qma-1}
completes the proof of Theorem~\ref{thm:qip-vs-qma}.


\section{Quantum Oracle Relative to Which \texorpdfstring{$\QSZKHV
    \not\subseteq \class{QCMA}$}{QSZK is not in QCMA}}\label{sec:oracle}

In order to prove the desired result we find a problem in $\QSZKHV$
and prove a black-box lower bound in the \class{QCMA} model.  We end
up with a quantum oracle, as the constructed problem makes essential
use of quantum information.  This approach is due to Aaronson and
Kuperberg~\cite{AaronsonK07}, who prove a similar result for
\class{QMA} versus \class{QCMA}.  The argument given here
is related to the argument of Aaronson and Kuperberg, both in structure and in the fact that we make use of a bound on the expected overlap of a state drawn from a $p$-uniform distribution with a fixed state.  The main difference is that in the problem we consider we need to extend the proof to the case where it is  a unitary operator that is hidden inside the oracle, not a pure state.
Note that subsequent to the completion of this work, Aaronson has
shown the stronger result that there is an oracle relative to which
$\class{SZK} \not\subseteq \class{QMA}$~\cite{Aaronson11}.

For our result we consider a black-box that takes as input a
control qubit, chooses a random pure state $\ket{\psi}$ and applies a
fixed but hidden $d$ by $d$ unitary $U$ to half of $\ket\psi$,
controlled by the input qubit. The hidden unitary $U$ can be inverted
by a $\class{QSZK}$ prover, but
in the \class{QCMA} model, the Verifier cannot invert $U$ and
recover the input with making an 
exponential number of queries to the black-box.
We prove a lower bound on the number of queries needed by a
\class{QCMA} Verifier to distinguish this black-box from one that
simply generates random pure states.
\begin{namedtheorem}{Theorem~\ref{thm:oracle-result2}}
  There exists a quantum oracle $A$ such that
  $\QSZKHV^A \not\subseteq \class{QCMA}^A$.
\end{namedtheorem}

\subsection{Background}\label{sec:measures}

Before proving the oracle result we review some background on 
measures on quantum states and channels that will be used in the proof.

Let $\unitary{H}$ be the group of unitary matrices acting on a
Hilbert space $\mathcal{H}$.  When no confusion is likely to arise, we
will also use the notation $\unitarydim{d}$, where $\dm{H} = d$.  The set of
pure states on $\mathcal{H}$, i.e.\ the unit sphere in $\mathcal{H}$,
is given by $\sphere{H}$ or $\spheredim{d-1}$.   We
refer to $d$-dimensional spaces for convenience: in general $d = 2^n$
for some space of $n$ qubits.

Throughout this section, the \emph{uniform} measure on states and
unitaries is given by the Haar measure.  In the case of unitaries, we
use $\mu_{\unitary{H}}$ to denote the Haar measure on the unitaries
on $\mathcal{H}$, that is, the unique left and right invariant measure
normalized so that $\mu_{\unitary{H}}(\unitary{H}) = 1$.  When
the space in question is clear we will drop the subscript and use only
$\mu$ to refer to this measure.  The Haar measure on $\sphere{H}$
can be obtained by applying a random $U \in \unitary{H}$ to a fixed
pure state (the invariance of the Haar measure implies that the choice
of the fixed state does not matter).  We will use
$\mu_{\sphere{H}}$ to refer to this measure.

Essential to our argument is the notion of a probability measure that
is \emph{nearly} uniform.  Following Aaronson and
Kuperberg~\cite{AaronsonK07}, given a measure $\sigma$ we say that it
is \emph{$p$-uniform} if $p \sigma \leq \mu$, where $\mu$ is the
uniform measure over the space in question.
This notion is directly related to the class \class{QCMA} by the fact
that if the verifier starts with a uniform measure and conditions on
a $m$-bit classical message, the result is a
$(2^{-m})$-uniform measure.
The main technical result of this section will be to show that
such a measure over $\unitarydim{d}$ does not help the verifier identify a
particular unitary, unless $m \in \Omega(d)$.
This result follows by a reduction to the pure state case, which
is the key to the quantum oracle that separates $\class{QMA}$ and
$\class{QCMA}$~\cite{AaronsonK07}.

Before doing this, we highlight two straightforward properties of $p$-uniform measures on $\unitarydim{d}$ and $\spheredim{d-1}$.
\begin{proposition}\label{prop:p-unif-properties}
  Let $\sigma$ be a $p$-uniform measure on $\unitarydim{d}$.
  \begin{enumerate}
    \item For any $U \in \unitarydim{d}$ the measure $U \sigma$ remains $p$-uniform.
      \label{item:p-unif-invar}
    \item For any $\ket\psi \in \spheredim{d-1}$, the measure $\tau$ on $\spheredim{d-1}$ given by
      \[ \tau( A ) = \sigma( \{ U : U \ket\psi \in A \} ) \]
      is $p$-uniform.
      \label{item:p-unif-states}
  \end{enumerate}
\end{proposition}

\begin{proof}
    The left-invariance of $\mu_{\unitarydim{d}}$ gives the first property,
    since for any $A \subseteq \unitarydim{d}$,
    \begin{align*}
      p (U \sigma)(A)
      = p \sigma( U^\dagger A )
      \leq \mu(U^\dagger A )
      = \mu(A).
    \end{align*}
    
    The second property follows from the definition of $\mu_{\spheredim{d-1}}$,
    \begin{align*}
     p \tau( A ) 
     = p \sigma( \{ U : U \ket\psi \in A \} ) 
     \leq \mu_{\unitarydim{d}} ( \{ U : U \ket\psi \in A \} ) 
     = \mu_{\spheredim{d-1}}(A).
    \end{align*}
    where right-invariance of $\mu_{\unitarydim{d}}$ implies
    that the choice of $\ket\psi$ does not matter.
   \end{proof}

\subsection{Oracle Separation}\label{sec:oracle-proof}

We now define our problem.

\begin{problem}\label{prob:oracle}
  Given a quantum oracle $O \colon \mathcal{A}
  \to \mathcal{A \tprod H \tprod K}$, where $\dm{H} = \dm{K} = d$ and
  $\dm{A}=2$.  The problem is to decide between the two cases
  \begin{enumerate}
  \item there exists a unitary $U \in \unitary{H}$ such that the oracle $O$ performs the map 
\begin{align*}     
\alpha \ket{0} + \beta \ket{1} 
      \mapsto
      \frac{1}{d^2} \Bigl( &
      \abs{\alpha}^2 \altketbra{0} \tprod \identity{H \tprod K}
		+ \alpha \bar\beta \ket{0}\bra{1} \tprod U^\dagger \tprod \identity{K}\\
     & + \bar\alpha \beta \ket{1}\bra{0} \tprod U \tprod \identity{K}
        +  \abs{\beta}^2 \altketbra{1} \tprod \identity{H \tprod K} 
        \Bigr).
\end{align*}
This map can be implemented in the following way: the oracle chooses a pure state $\ket\psi \in \mathcal{H \tprod K}$ from the Haar measure and then performs the map
    \begin{equation*}
      \alpha \ket 0 + \beta \ket 1 \mapsto \alpha \ket 0 \ket\psi + \beta \ket 1 (U \tprod \identity{K}) \ket\psi.
   \end{equation*}
        \label{item:oracle-type-1}
  \item the oracle $O$ preforms the map
\[    
\alpha \ket{0} + \beta \ket{1} 
      \mapsto
      \frac{1}{d^2} \left( 
        \abs{\alpha}^2 \altketbra{0} \tprod \identity{H \tprod K}
		+ \abs{\beta}^2 \altketbra{1} \tprod \identity{H \tprod K}
        \right).
\]
for example by measuring the input qubit and appending the maximally mixed state.
      \label{item:oracle-type-2}
	\end{enumerate}
  \end{problem}
  
We defined the oracles as superoperators, but one can think of them as
unitaries in larger spaces.
The key idea is that in the first case the coherence of the input
qubit can be recovered, provided the hidden unitary $U$ can be
inverted, whereas in the second case this coherence is irretrievably
lost.
The prover in a \class{QSZK}
protocol, given only the portion of the state in the space $\mathcal{H}$
and a copy of the input qubit, is able to apply $U^\dagger$ in order to disentangle the
input space from $\mathcal{H \tprod K}$.  
To prove a lower bound on this problem, we argue that with at most a
small amount of knowledge about the hidden operator $U$, an oracle of
the first type appears much the same as an oracle of the second type.

Before proving this lower bound, we give an 
interactive protocol
for the problem.  The idea behind the protocol is that when the input
to the oracle is one half of a maximally entangled state then in the
first case a prover is able to assist the verifier in recovering the
original input state, but in the second case no action of the prover
can recover the state.
\begin{proto}\label{proto:oracle-qszk}
  Let $O$ be the oracle in Problem~\ref{prob:oracle}.
  \begin{enumerate}
  \item $V$, prepares the state $\ket{\phi^+} = (\ket{00} +
    \ket{11})/\sqrt{2} \in \mathcal{B \tprod A}$, and uses as input to
    the oracle $O$ the portion of the state in $\mathcal{A}$.  $V$
    then sends the state in $\mathcal{A \tprod H}$ to $P$.
  \item P applies the unitary $U^\dagger$ on $\spa{H}$ controlled on the qubit in $\spa{A}$.  
  \item $V$ receives a state from $P$ in the space $\mathcal{A \tprod
      H}$ and measures the operator $\altketbra{\phi^+}$ on the space
    $\mathcal{B \tprod A}$, accepting if and only if the outcome is one.
    \label{item:proto-step-2}
  \end{enumerate}
\end{proto}

In the following theorem we prove the completeness and soundness of Protocol~\ref{proto:oracle-qszk}. The fact that it is also zero-knowledge is argued as part of the proof of  Theorem~\ref{thm:oracle-result}.

\begin{theorem}\label{thm:qszk-containment}
  Let $V$ be the verifier in Protocol~\ref{proto:oracle-qszk}.
  \begin{enumerate}
  \item If the oracle is of type~\ref{item:oracle-type-1}, there is a
    prover $P$ that causes $V$ to accept with certainty.
    \label{item:thm-proto-completeness}
  \item If the oracle is of type~\ref{item:oracle-type-2}, then for
    any $P$, $V$ accepts with probability at most 1/2.
    \label{item:thm-proto-soundness}
 \end{enumerate}
  
  \begin{proof}
    To prove completeness (item~\ref{item:thm-proto-completeness}), 
    notice that when the oracle is of
    type~\ref{item:oracle-type-1}, the state of the verifier before
    sending the message to the prover is
   \[      \frac{1}{2 d^2} \left[ \altketbra{00} \tprod \identity{H \tprod K}
         + \ket{00}\bra{11} \tprod U^\dagger \tprod \identity{K}
         + \ket{11}\bra{00} \tprod U \tprod \identity{K}
         +  \altketbra{11} \tprod \identity{H \tprod K}. \right]
    \]
    If the honest prover applies
    $U^\dagger$ on the space $\mathcal{H}$, controlled on the qubit in
    $\mathcal{A}$, the state of the verifier at the start of
    Step~\ref{item:proto-step-2} is
    \[   
      \frac{1}{2 d^2} \bigl( \altketbra{00}
      + \ket{00}\bra{11}
      + \ket{11}\bra{00}
      + \altketbra{11} \bigr)
        \tprod \identity{H \tprod K}
        = \ket{\phi^+}\bra{\phi^+} \otimes \frac{\identity{H \tprod K}}{d^2}
    \]
   and so the projective measurement on $\mathcal{A \tprod B}$ given
    by $\{\altketbra{\phi^+}, \id - \altketbra{\phi^+} \}$ always
    results in the first outcome.  This implies that the verifier can
    always be made to accept an oracle of
    type~\ref{item:oracle-type-1}.

    To prove soundness (item~\ref{item:thm-proto-soundness})
    we show that the verifier rejects an oracle of
    type~\ref{item:oracle-type-2} with probability at least $1/2$,
    regardless of the strategy of the prover.  In this case the state
    of the verifier before sending the message is given by the mixture
    \begin{equation*}
      \frac{1}{2 d^2} \left( \altketbra{00} \tprod \identity{H \tprod K}
      + \altketbra{11} \tprod \identity{H \tprod K} \right).
    \end{equation*}
    After the prover applies an arbitrary transformation to 
    $\mathcal{A \tprod H}$, the result is
    \begin{equation*}
      \frac{1}{2 d} \left( \altketbra{0} \tprod \rho_0 \tprod \identity{K}
      +\altketbra{1} \tprod \rho_1 \tprod \identity{K} \right)
    \end{equation*}
    for some mixed states $\rho_0, \rho_1$ on $\mathcal{A \tprod H}$.
    The probability that the verifier's measurement results in the outcome
    $\altketbra{\phi^+}$ on this state is given by
    \begin{align*}
      \frac{1}{2d} \tr \left[ \altketbra{\phi^+}         
         \left( \altketbra{0} \tprod \rho_0 \tprod \identity{K}
          + \altketbra{1} \tprod \rho_1 \tprod \identity{K}
        \right) \right]
      = \frac{1}{4} \left( \bra 0 \rho_0 \ket 0 + \bra 1 \rho_1 \ket 1 \right)
      \leq \frac{1}{2},
    \end{align*}
    which implies that the verifier accepts with probability at most
    $1/2$ when $O$ is of type~\ref{item:oracle-type-2}. 
    In fact, the best strategy for a cheating prover is not to change
    the control bit in $\spa{A}$ at all.
  \end{proof}
\end{theorem}

A central component of the argument that a \class{QCMA} verifier
cannot identify a pure state hidden in an oracle is a geometric bound
on the expected overlap between any fixed state and a state drawn from
a $p$-uniform distribution.

\begin{lemma}[Aaronson and Kuperberg~\cite{AaronsonK07}]
  \label{lem:ak07}
  For any $p$-uniform measure $\sigma$ on $\spheredim{d-1}$ and any state $\rho$
  \begin{equation*}
    \expect_{\ket\psi \in \sigma} \left[ \bra{\psi} \rho \ket{\psi} \right]
    \in O\left( \frac{1 + \log 1/p}{d} \right)
  \end{equation*}
\end{lemma}

Our argument requires a similar geometric bound, except that we have a
$p$-uniform measure over unitaries and not the pure states.  We obtain
a reduction from $\unitarydim{d}$ to $\spheredim{d-1}$, which allows us to extend the
bound in Lemma~\ref{lem:ak07}.

\begin{lemma}\label{lem:expectation-bound}
  If $\sigma$ is a $p$-uniform measure on $\unitarydim{d}$, then
  \begin{equation*}
    \tnorm{ \expect_{U \in \sigma} U } \in O\left( \sqrt{ d (1 + \log 1/p) } \right)
  \end{equation*}
\end{lemma}
\begin{proof}
    Let $\sigma$ be an arbitrary $p$-uniform measure, then
    \begin{equation*}
      \tnorm{ \expect_{U \in \sigma} \left[ U \right] }
      = \max_{V \in \unitarydim{d}} \abs{ \tr \expect_{U \in \sigma} \left[  U \right]  V }
      = \max_{V \in \unitarydim{d}} \abs{ \expect_{U \in \sigma} \left[ \tr UV  \right] }
      = \max_{V \in \unitarydim{d}} \abs{ \expect_{U \in \sigma V} \left[ \tr U \right]}.
    \end{equation*}
    Notice however that the measure $\sigma V$ is $p$-uniform whenever
    $\sigma$ is, and so by Proposition~\ref{prop:p-unif-properties} we
    may, since $\sigma$ is arbitrary,
    discard the maximization over $V$.
    Doing so, the desired quantity is
    \begin{equation}\label{eq:state-reduction}
      \abs{\expect_{U \in \sigma} \tr{U} }
      \leq \expect_{U \in \sigma} \abs{ \tr{U} }
      = \expect_{U \in \sigma} \sum_{i=1}^d  \abs{ \bra i U \ket i }
      = \sum_{i=1}^d  \expect_{\ket{\psi_i} \in \tau_i}  \abs{ \braket{i}{\psi_i} },
    \end{equation}
    where for each $i$, $\tau_i$ is the $p$-uniform 
    measure on $\spheredim{d-1}$ obtained
    by applying a $\sigma$-distributed unitary $U$ to the state
    $\ket{i}$.  Having reduced the problem to an expectation over a
    $p$-uniform measure on pure states, we apply the bound in
    Lemma~\ref{lem:ak07} to Equation~\eqref{eq:state-reduction} to
    get
    \begin{equation*}
      \tnorm{ \expect_{U \in \sigma} \left[ U \right] }
      \leq  \sum_{i=1}^d  O\left( \sqrt{\frac{1 + \log 1/p}{d} } \right)
      =  O\left( \sqrt{ d (1 + \log 1/p) } \right),
    \end{equation*}
    as in the statement of the Lemma.
  \end{proof}

\begin{theorem}\label{thm:lower-bound}
  Any \class{QCMA} protocol for problem~\ref{prob:oracle} with an
  $m$-bit witness uses $\Omega(\sqrt{d/(m+1)})$ calls to the oracle.
\end{theorem}
\begin{proof}
   Consider any \class{QCMA} protocol with any $m$-bit witness.  We
    will show that this protocol requires at least
    $\Omega(\sqrt{d/(m+1)})$ calls to the oracle to determine whether
    it is an oracle of the first or second type.

    We use the hybrid approach of Bennet et al.~\cite{BennettB+97}.
    Let $\rho_0$ be the initial state of the algorithm.  Let
    $\rho_i$ be the state of the algorithm immediately after
    the $i$th call to an oracle of type~\ref{item:oracle-type-2}.  After $T$ calls to such an
    oracle, we denote the final state of the algorithm (before the
    measurement of whether or not to accept) as $\rho_T$.  In the case
    that the algorithm is run on an oracle of type 1, we denote the
    final state by $\xi_T$.  
   Our goal is to show that the distance between $\rho_{T}$ and
    $\xi_{T}$ is small, unless $T$, the number of oracle calls, is
    sufficiently large.  We will do this by considering running the
    algorithm for $(i-1)$ queries on an oracle of type~\ref{item:oracle-type-2} and then
    switching the oracle to type~\ref{item:oracle-type-1}.  We denote the state obtained in
    this way by $\rho_i'$.
    We prove that this state is very close to the state $\rho_i$,
    which will give the desired result, since $\tnorm{\xi_T - \rho_T}
    \leq \sum_{i=1}^T \tnorm{\rho_i - \rho_i'}$ by the triangle
    inequality.

    Let $\ket\nu = \alpha \ket 0 + \beta \ket 1$ and let $\nu =
    \altketbra\nu$ be the input to the $(k+1)$st call to the oracle,
    after the algorithm has been run for $k$ queries to an oracle of
    type 2.  Strictly speaking, $\nu$ may be mixed state, but a
    convexity argument implies that a pure input state will maximize
    the distance between the output states of the two oracles.
    The output of the
    $O_2$ on the pure state $\nu$ is the mixed state
    \begin{align}
      O_2(\nu) 
      &= \frac{1}{d^2} \left( \abs{\alpha}^2 \altketbra{0} \tprod \identity{H \tprod K}
        +  \abs{\beta}^2 \altketbra{1} \tprod \identity{H \tprod K} \right).
     \label{eq:o2-mixture}
    \end{align}
    The output of the oracle $O_1$, for a fixed hidden unitary $U$, is
\begin{align*}
      O_1^U(\nu)   =  \frac{1}{d^2} \Bigl( &
        \abs{\alpha}^2 \altketbra{0} \tprod \identity{H \tprod K}
		+ \alpha \bar\beta \ket{0}\bra{1} \tprod U^\dagger \tprod \identity{K}
   + \bar\alpha \beta \ket{1}\bra{0} \tprod U \tprod \identity{K}
        +  \abs{\beta}^2 \altketbra{1} \tprod \identity{H \tprod K} \Bigr).
\end{align*}
   However, since this is the first query the algorithm has made to
    the oracle $O_1$, it has no information about the hidden unitary $U$, except the
    $m$-bit classical message from the \class{QCMA} prover.  This
    information constrains the unitary $U$ to a
    $2^{-m}$-uniform distribution $\sigma$, so that the output of oracle
    $O_1$ can be represented by the mixture of
    the previous equation over all $U \in \sigma$, which is
    \begin{align}
      O_1(\nu) &= \expect_{U \in \sigma} \left[ O_1^U(\nu) \right]
      \label{eq:o1-mixture}
   \end{align}
One way to think about this, is that the oracle $O_1$ has another
space which is initialized to be a uniform superposition of
descriptions of all possible unitaries. Then the oracle uses this
register as a control in order to apply the mapping $O^U_1$. 
The classical \class{QCMA} message could be thought of as an outcome to a
partial measurement on this register, which resulted in the collapse
of the uniform superposition to a $p$-uniform superposition of the
unitaries consistent with the measurement outcome. 
The verifier's view can be calculated by tracing out this register.

  The remaining task is to compute the diamond norm of the difference
   of Equations~\eqref{eq:o2-mixture} and~\eqref{eq:o1-mixture}, which
   will measure the maximum probability that any measurement can
   distinguish whether or not a single call to the oracle $O_1$ has been
   replaced by a call to $O_2$.  
 \begin{equation*}
     \tnorm{O_1(\nu) -  O_2(\nu)}
     =  \frac{1}{d^2}
  \tnorm{ \alpha \bar\beta \ket{0}\bra{1} \tprod \expect_{U \in \sigma}[U^\dagger \tprod \identity{K}] + \bar\alpha \beta \ket{1}\bra{0} \tprod \expect_{U \in \sigma}[U \tprod \identity{K}]} 
   \end{equation*}

   We then use the fact that $\tnorm{\ket 0 \bra 1 \tprod A^\dagger +
     \ket 1 \bra 0 \tprod A} = 2 \tnorm{A}$ (see~\cite[Section
   II.1]{Bhatia97} for the relationship between the eigenvalues of an
   operator of this form and the singular values of $A$).  This
   implies that
   \begin{equation*}
     \tnorm{O_1(\nu) -  O_2(\nu)}
     = \frac{2 \abs{\alpha} \abs{\beta}}{d^2} \tnorm{ \expect_{U \in \sigma} [U \tprod \id ]}
     = \frac{ 2 \abs{\alpha} \abs{\beta} }{d} \tnorm{ \expect_{U \in \sigma} [U] }.
     \label{eq:oracle-diff-1}
   \end{equation*}
   Finally, since $\sigma$ is a $2^{-m}$ uniform measure on $\unitarydim{d}$ we
   apply Lemma~\ref{lem:expectation-bound} to obtain
   \begin{equation}\label{eq:oracle-one-round-small}
     \tnorm{O_1(\nu) - O_2(\nu)} 
     \in O\left( \sqrt{ \frac{1 + m}{d}} \right).
   \end{equation}    

   This equation bounds the trace distance of the output states of the
   two oracles.  The maximum distance between the
   states $\rho_i$ and $\rho_i'$ is upper bounded by the diamond norm, which
   takes into account the fact that the algorithm may use an ancillary
   space to better distinguish the two oracles.  Using the fact that
   the diamond norm of the difference of two channels is achieved by a
   pure quantum state~\cite{RosgenW05}, we have shown that there exists some
   pure state $\nu$ such that for all $i \in \{1, \ldots, T\}$
   \begin{align*}
     \tnorm{\rho_i - \rho_i'}
     \leq \dnorm{O_1 - O_2}
     \leq 2 \tnorm{O_1(\nu) - O_2(\nu)}
     \in O\left( \sqrt{ (1 + m)/d } \right),
   \end{align*}
   where we have used Lemma~\ref{thm:operator-norm-dim-bound} to
   upper bound the diamond norm by the trace norm.
   The triangle inequality implies that replacing all $T$ calls to $O_1$
   with calls to $O_2$ results in states $\rho_T$ and $\xi_T$ with
   trace distance
   \begin{equation*}
     \tnorm{\rho_T - \xi_T} 
     \leq \sum_{i=1}^T \tnorm{\rho_i - \rho_i'}
     \in  O\left( T \sqrt{ (1 + m)/d}  \right). 
   \end{equation*}
   This implies that in order for a black-box algorithm to
   distinguish $O_1$ and $O_2$ with constant probability it is
   required to make $T = \Omega( \sqrt{ d / (1+m) })$ calls to the oracle.
 \end{proof}

We now use Protocol~\ref{proto:oracle-qszk} and the lower bound in
Theorem~\ref{thm:lower-bound} to obtain an oracle relative to which
\class{QSZK} is not contained in \class{QCMA}.  The proof of this
follows very closely the argument of Aaronson and
Kuperberg~\cite{AaronsonK07}, who establish an oracle relative to
which \class{QMA} is not in \class{QCMA}.

Strictly speaking, we find a quantum oracle $A$ such that $\QSZKHV^A
\not\subseteq \class{QCMA}^A$, i.e.\ we deal only with the honest
verifier case.  While it is known that $\QSZKHV =
\class{QSZK}$~\cite{Watrous09zero-knowledge}, we do not know if this
is still the case given access to the oracle $A$.

\begin{namedtheorem}{Theorem~\ref{thm:oracle-result}}
  There exists a quantum oracle $A$ such that
  $\QSZKHV^A \not\subseteq \class{QCMA}^A$
\end{namedtheorem}

  \begin{proof}
    Let $L$ be a random unary
    language that we will use to define the oracle $A = \{A_n\}$.  For
    each $n$, $A_n$ takes $2n$ qubits as input (so that $d = 2^n$ in
    Problem~\ref{prob:oracle}).  For each $n$ there are two cases.  If
    $1^n \in L$ then $A_n$ is an oracle of
    type~\ref{item:oracle-type-1} in Problem~\ref{prob:oracle}, i.e.\
    $A_n$ implements some hidden unitary $U$ on half of the input
    qubits.  On the other hand, if $1^n \not\in L$, then $A_n$ is of
    type~\ref{item:oracle-type-2}.

    We use Theorem~\ref{thm:qszk-containment} to give an honest-verifier
    \class{QSZK} protocol for $L$, given access to the oracle $A$.
    For a given input $1^n$, the Verifier first runs
    protocol~\ref{proto:oracle-qszk} to determine the type of the
    oracle.  The verifier accepts that $1^n \in L$ if and only if this
    protocol accepts. The completeness and soundness of the protocol have  
    already been shown.
    Last, it is easy to show that the protocol is zero knowledge for
    the honest verifier.
    The state of the verifier after Step~\ref{item:oracle-type-1} 
    can be simulated by the
    simulator, since it has at its disposal both the honest verifier
    and the oracle.
    After the prover's message, in the `yes' case, the state is equal to
    \[ \ket{\phi^+}\bra{\phi^+} \otimes \identity{H \tprod K} /d^2 \]
    which can also be easily simulated, and so the protocol is
    (honest-verifier) zero-knowledge.
    This implies that $L \in \QSZKHV^A$.

    We then use the lower bound in Theorem~\ref{thm:lower-bound} to
    show that $L \not\in \class{QCMA}^A$, with probability one (over
    the choice of $L$ and the hidden unitary $U$ in the oracle).  
    This portion of the proof is identical to the proof in~\cite{AaronsonK07}, but for clarity we repeat it here.  
    Fix $M$ an arbitrary $\class{QCMA}$
    verifier and let $S_M(n)$ represent the event that the verifier
    $M$ succeeds on the input $1^n$, i.e.\ either $1^n \in L$ and
    there exists a witness string $w$ such that $M^A$ accepts with
    probability at least $2/3$, or $1^n \not\in L$ and no witness $w$
    causes $M$ to accept with probability larger than $1/3$.
    Theorem~\ref{thm:lower-bound} implies that $M$ fails for large
    enough $n$, i.e.\ that for some $N$ it holds that for all $n \geq
    N$
    \begin{equation*}
      \Pr_{L,V} [ S_M(n) | S_M(1), \ldots, S_M(n-1) ] \leq \frac{2}{3}.
    \end{equation*}
    This implies that the probability that $M$ works on all $n$ is 0,
    i.e.\
    \begin{equation*}
      \Pr_{L,V} [ S_M(1) \wedge S_M(2) \cdots ] = 0.
    \end{equation*}
    Finally, since there are only a countably infinite number of
    \class{QCMA} verifiers (by the Solovay-Kitaev Theorem~\cite{Kitaev97}),
    the union bound implies that with probability one we have $L \not\in \class{QCMA}$.
  \end{proof}

\subsection*{Acknowledgements}

BR is supported by the Centre for Quantum Technologies,
which is funded by the Singapore Ministry of Education and the
Singapore National Research Foundation.
AC and IK are supported by projects  ANR-09-JCJC-0067-01, ANR-08-EMER-012 and QCS (grant 255961) of the E.U.

\newcommand{\arxiv}[2][quant-ph]{\href{http://arxiv.org/abs/#2}{arXiv:#2 [#1]}}
  \newcommand{\oldarxiv}[2][quant-ph]{\href{http://arxiv.org/abs/#1/#2}{arXiv:#1/#2}}

\end{document}